\newif\ifarxiv\arxivtrue%

\documentclass[a4paper,UKenglish,cleveref,numberwithinsect]{lipics-v2019}

\nolinenumbers%
\ArticleNo{41}
\EventAcronym{CONCUR}
\EventYear{2019}

\usepackage{silence}

\usepackage{microtype}
\usepackage{xspace}
\usepackage{enumerate}

\usepackage{amsmath}
\usepackage{amssymb}
\usepackage{mathpartir}
\usepackage{stmaryrd}

\newcommand{\pipe}{\;\;|\;\;}

\newcommand{\BA}{\ensuremath{\mathsf{BA}}\xspace}
\newcommand{\KA}{\ensuremath{\mathsf{KA}}\xspace}
\newcommand{\CKA}{\ensuremath{\mathsf{CKA}}\xspace}
\newcommand{\KAT}{\ensuremath{\mathsf{KAT}}\xspace}
\newcommand{\CKAT}{\ensuremath{\mathsf{CKAT}}\xspace}
\newcommand{\WKAO}{\ensuremath{\mathsf{WKAO}}\xspace}
\newcommand{\KAO}{\ensuremath{\mathsf{KAO}}\xspace}
\newcommand{\AKA}{\ensuremath{\mathsf{T}}\xspace}

\newcommand{\terms}{\ensuremath{\mathcal{T}}}
\newcommand{\termsprop}{\terms_P}
\newcommand{\termsrat}{\terms_R}
\newcommand{\termsgrat}{\terms_{GR}}
\newcommand{\termssrat}{\terms_{SR}}
\newcommand{\termsgsrat}{\terms_{GSR}}
\newcommand{\termsagrat}{\terms_{AGR}}

\newcommand{\at}{\mathcal{A}}

\newcommand{\sem}[1]{\llbracket#1\rrbracket}
\newcommand{\semba}[1]{\sem{#1}_{\scriptscriptstyle\BA}}
\newcommand{\semka}[1]{\sem{#1}_{\scriptscriptstyle\KA}}

\newcommand{\semwkao}[1]{\sem{#1}_{\scriptscriptstyle\WKAO}}
\newcommand{\semkao}[1]{\sem{#1}_{\scriptscriptstyle\KAO}}

\newcommand{\bigsemwkao}[1]{{\left\llbracket#1\right\rrbracket}_{\scriptscriptstyle\WKAO}}
\newcommand{\bigsemkao}[1]{{\left\llbracket#1\right\rrbracket}_{\scriptscriptstyle\KAO}}

\newcommand{\equivba}{\equiv_{\scriptscriptstyle\BA}}
\newcommand{\equivka}{\mathrel{\equiv_{\scriptscriptstyle\KA}}}
\newcommand{\equivcka}{\mathrel{\equiv_{\scriptscriptstyle\CKA}}}
\newcommand{\equivkat}{\mathrel{\equiv_{\scriptscriptstyle\KAT}}}
\newcommand{\equivckat}{\mathrel{\equiv_{\scriptscriptstyle\CKAT}}}
\newcommand{\equivwkao}{\mathrel{\equiv_{\scriptscriptstyle\WKAO}}}
\newcommand{\equivkao}{\mathrel{\equiv_{\scriptscriptstyle\KAO}}}
\newcommand{\equivaka}{\mathrel{\equiv_{\scriptscriptstyle\AKA}}}

\newcommand{\leqqba}{\leqq_{\scriptscriptstyle\BA}}
\newcommand{\leqqka}{\mathrel{\leqq_{\scriptscriptstyle\KA}}}
\newcommand{\leqqcka}{\mathrel{\leqq_{\scriptscriptstyle\CKA}}}
\newcommand{\leqqkat}{\mathrel{\leqq_{\scriptscriptstyle\KAT}}}
\newcommand{\leqqckat}{\mathrel{\leqq_{\scriptscriptstyle\CKAT}}}

\newcommand{\leqqkao}{\mathrel{\leqq_{\scriptscriptstyle\KAO}}}
\newcommand{\leqqaka}{\mathrel{\leqq_{\scriptscriptstyle\AKA}}}

\newcommand{\contr}{\mathrel{\preceq}}
\newcommand{\altcontr}{\mathrel{\trianglelefteqslant}}
\newcommand{\down}[1]{#1{\downarrow}}

\newcommand{\naturals}{\mathbb{N}}

\usepackage{thm-restate}
\input{fixrestatable}

\let\OLDthebibliography\thebibliography%
\renewcommand\thebibliography[1]{
  \OLDthebibliography{#1}
  \setlength{\itemsep}{-0.2pt}
}

\theoremstyle{plain}
\declaretheorem[name=Theorem,sibling=theorem]{thrm}
\declaretheorem[name=Lemma,sibling=theorem]{lem}

\declaretheorem[name=Antinomy,sibling=theorem]{ant}

\crefname{lem}{Lemma}{Lemmas}
\crefname{prop}{Proposition}{Propositions}

\title{%
    Kleene Algebra with Observations%
}

\author{Tobias Kapp\'{e}}{University College London}{tkappe@cs.ucl.ac.uk}{https://orcid.org/0000-0002-6068-880X}{} 
\author{Paul Brunet}{University College London}{}{https://orcid.org/0000-0002-9762-6872}{} 
\author{Jurriaan Rot}{University College London \and Radboud University, Nijmegen}{}{}{Marie Curie Fellowship (795119).}
\author{Alexandra Silva}{University College London}{}{https://orcid.org/0000-0001-5014-9784}{Leverhulme Prize (PLP--2016--129).} 
\author{Jana Wagemaker}{University College London}{}{}{}
\author{Fabio Zanasi}{University College London}{}{}{EPSRC grant (EP/R020604/1).}

\funding{ERC Starting Grant ProFoundNet (679127).}

\authorrunning{T. Kapp\'{e}, P. Brunet, J. Rot, A. Silva, J. Wagemaker and F. Zanasi}

\Copyright{T. Kapp\'{e}, P. Brunet, J. Rot, A. Silva, J. Wagemaker and F. Zanasi}

\ccsdesc{Theory of computation~Formal languages and automata theory}

\relatedversion{A full version of the paper is available at \url{https://arxiv.org/abs/1811.10401}}

\acknowledgements{We are grateful to Jean-Baptiste Jeannin, Dan Frumin and Damien Pous individually, for their input which helped contextualise this work.}

\begin{document}

\maketitle

\begin{abstract}
\emph{Kleene algebra with tests} (\emph{KAT}) is an algebraic framework for reasoning about the control flow of sequential programs.
Generalising KAT to reason about concurrent programs is not straightforward, because axioms native to KAT in conjunction with expected axioms for concurrency lead to an anomalous equation.
In this paper, we propose \emph{Kleene algebra with observations} (\emph{KAO}), a variant of KAT, as an alternative foundation for extending KAT to a concurrent setting.
We characterise the free model of KAO, and establish a decision procedure w.r.t.\ its equational theory.

\keywords{Concurrent Kleene algebra, Kleene algebra with tests, free model, axiomatisation, decision procedure}
\end{abstract}

\section{Introduction}

The axioms of \emph{Kleene algebra} (\emph{KA})~\cite{kleene-1956,conway-1971} correspond well to program composition~\cite{hoare-hayes-he-etal-1987}, making them a valuable tool for studying equivalences between programs from an algebraic perspective.
An extension of Kleene algebra known as \emph{Kleene algebra with tests} (\emph{KAT})~\cite{kozen-1996} adds primitives for conditional branching, and is particularly useful when proving validity of program transformations, such as optimisations applied by a compiler~\cite{kozen-patron-2000,smolka-eliopoulos-foster-guha-2015}.

As a matter of fact, KAT is sufficiently abstract to express not only program behaviour, but also program specifications; consequently, its laws can be used to compare programs to specifications~\cite{kozen-2000,anderson-foster-guha-etal-2014}.
What makes this connection especially powerful is that KA (resp.\ KAT) is known to be \emph{sound} and \emph{complete} with respect to a language model~\cite{boffa-1990,krob-1991,kozen-1994,kozen-smith-1996}, meaning that an equation is valid in any KA (resp.\ KAT) precisely when it holds in the corresponding language model.
Practical algorithms for deciding language equivalence~\cite{hopcroft-karp-1971,bonchi-pous-2013,pous-2014} enable checking equations in KA or KAT, and hence automated verification becomes feasible~\cite{foster-kozen-milano-silva-thompson-2015}.

More recently, Kleene algebra has been extended with a parallel composition operator, yielding \emph{concurrent Kleene algebra} (\emph{CKA})~\cite{hoare-moeller-struth-wehrman-2009,hoare-2014,hoare-staden-moeller-struth-zhu-2016}.
Crucially, CKA includes the \emph{exchange law}, which encodes \emph{interleaving}, i.e., the (partial) sequentialisation of threads.
Like its predecessors, CKA can be applied to verify (concurrent) programs by reasoning about equivalences~\cite{hoare-staden-moeller-struth-zhu-2016}.
The equational theory of CKA has also been characterised in terms of a language-like semantics~\cite{laurence-struth-2017-arxiv,kappe-brunet-silva-zanasi-2018}, where equivalence is known to be decidable~\cite{brunet-pous-struth-2017}.

Since both KAT and CKA are conservative extensions of KA, this prompted Jipsen and Moshier~\cite{jipsen-moshier-2016} to study a marriage between the two, dubbed \emph{concurrent Kleene algebra with tests} (\emph{CKAT}).
The aim of CKAT is to extend CKA with Boolean guards, and thus arrive at a new algebraic perspective on verification of concurrent programs with conditional branching.

The starting point of this paper is the realisation that CKAT is not a suitable model of concurrent programs.
This is because for any test $p$ and CKAT-term $e$, one can prove $p \cdot e \cdot \overline{p} \equivckat 0$, an equation that appears to have no reasonable interpretation for programs.
The derivation goes as follows:
\begin{equation*}
0 \leqqkat
 	    p \cdot e \cdot \overline{p}
        \leqqcka e \parallel (p \cdot \overline{p})
        \equivkat e \parallel 0
        \equivcka 0
        \ .
\end{equation*}
As we shall see, this is possible because of the interplay between the exchange law and the fact that KAT identifies conjunction of tests with their sequential composition.
For sequential programs, this identification is perfectly reasonable.
In the context of concurrency with interleaving, however, actions from another thread may be scheduled in between two sequentially composed tests, whereas the conjunction of tests executes atomically.
Indeed, an action scheduled between the tests might very well change the result of the second test.

It thus appears that, to reason algebraically about programs with both tests and concurrency, one needs a perspective on conditional branching where the conjunction of two tests is not necessarily the same as executing one test after the other.
The remit of this paper is to propose an alternative to KAT, which we call \emph{Kleene algebra with observations} (\emph{KAO}), that makes exactly this distinction.
We claim that, because of this change, KAO is more amenable to a sensible extension with primitives for concurrency.
Establishing the meta-theory of KAO turns out to be a technically demanding task.
We therefore devote this paper to such foundations, and leave development of concurrent KAO to follow-up work.

Concretely, we characterise the equational theory of KAO in terms of a language model (\cref{section:completeness}).
Furthermore, we show that we can decide equality of these languages (and hence the equational theory of KAO) by deciding language equivalence of non-deterministic finite automata (\cref{section:decision-procedure}).
Both proofs show a clear separation of concerns: their kernel is idiomatic to KAO, and some well-known results from KA complete the argument.

\medskip

For space reasons, detailed proofs are only included in the full version of this paper~\cite{kao-full}; here, we sketch the main insights needed to prove the core propositions and theorems.

\section{Preliminaries}%
\label{section:preliminaries}

We start by outlining some concepts and elementary results.

\subparagraph*{Boolean algebra}
We use $2$ to denote the set $\{ 0, 1 \}$.
The \emph{powerset} (i.e., set of subsets) of a set $S$ is denoted $2^S$.
We fix a finite set $\Omega$ of symbols called \emph{observables}.

The \emph{propositional terms} over $\Omega$, denoted $\termsprop$, are generated by the grammar
\[
    p, q ::= \bot \pipe \top \pipe o \in \Omega \pipe p \vee q \pipe p \wedge q \pipe \overline{p} \ .
\]
We write $\equivba$ for the smallest congruence on $\termsprop$ that satisfies the axioms of Boolean algebra, i.e., such that for all $p, q, r \in \termsprop$ the following hold:
\begin{mathpar}
p \vee \bot \equivba p
\and
p \vee q \equivba q \vee p
\and
p \vee \overline{p} \equivba \top
\and
p \vee (q \vee r) \equivba (p \vee q) \vee r
\\
p \wedge \top \equivba p
\and
p \wedge q \equivba q \wedge p
\and
p \wedge \overline{p} \equivba \bot
\and
p \wedge (q \wedge r) \equivba (p \wedge q) \wedge r
\\
p \vee (q \wedge r) \equivba (p \vee q) \wedge (p \vee r)
\and
p \wedge (q \vee r) \equivba (p \wedge q) \vee (p \wedge r)
\ .
\end{mathpar}
The set of \emph{atoms}, denoted $\at$, is defined as $2^\Omega$.
The semantics of propositional terms is given by the map $\semba{-}: \termsprop \to 2^{\at}$, as follows:
\begin{align*}
\semba{\bot} &= \emptyset
    & \semba{o} &= \{ \alpha \in \at : o \in \alpha \}
    & \semba{p \vee q} &= \semba{p} \cup \semba{q} \\
\semba{\top} &= \at
    & \semba{\overline{p}} &= \at \setminus \semba{p}
    & \semba{p \wedge q} &= \semba{p} \cap \semba{q}
      \ .
\end{align*}
We also write $p \leqqba q$ as a shorthand for $p \vee q \equivba q$.

It is known that $\semba{-}$ characterises $\equivba$ (c.f.~\cite[Chapter~5.9]{birkhoff-bartee-1970}), in the following sense:
\begin{thrm}[Completeness for $\BA$]%
\label{theorem:ba-completeness}
Let $p, q \in \termsprop$; now $p \equivba q$ if and only if $\semba{p} = \semba{q}$.
\end{thrm}

When $\alpha \in \at$, we write $\pi_\alpha$ for the Boolean term $\bigwedge_{o \in \alpha} o \wedge \bigwedge_{o \in \Omega \setminus \alpha} \overline{o}$, in which $\bigwedge$ is the obvious generalisation of $\wedge$ for some (arbitrary) choice of bracketing and order on terms.
The following is then straightforward to prove.
\begin{lem}%
\label{lemma:atom-semantics}
For all $\alpha \subseteq \Omega$ it holds that $\semba{\pi_\alpha} = \{ \alpha \}$.
\end{lem}

\subparagraph*{Kleene algebra}
A \emph{word} over a set $\Delta$ is a sequence of symbols $d_0\cdots{}d_{n-1}$ from $\Delta$.
The \emph{empty word} is denoted $\epsilon$.
A set of words is called a \emph{language}.
Words can be \emph{concatenated}: if $w$ and $x$ are words, then $wx$ is the word where the symbols of $w$ precede those of $x$.
If $L$ and $L'$ are languages over $\Delta$, then $L \cdot L'$ is the language of pairwise concatenations from $L$ and $L'$, i.e., $\{ wx : w \in L, x \in L' \}$.
We write $L^\star$ for the \emph{Kleene closure} of $L$, which is the set $\{ w_0\cdots{}w_{n-1} : w_0, \dots, w_{n-1} \in L \}$.
This makes $\Delta^\star$ the set of all words over $\Delta$.%

We fix a finite set of symbols $\Sigma$ called the \emph{alphabet}.
The \emph{rational terms} over $\Sigma$, denoted $\termsrat$, are generated by the grammar
\[
    e, f ::= 0 \pipe 1 \pipe a \in \Sigma \pipe e + f \pipe e \cdot f \pipe e^\star \ .
\]
We write $\equivka$ for the smallest congruence on $\termsrat$ that satisfies the axioms of Kleene algebra, i.e., such that for all $e, f, g \in \termsrat$ the following hold:
\begin{mathpar}
e + 0 \equivka e \and
e + e \equivka e \and
e + f \equivka f + e \and
e + (f + g) \equivka (e + f) + g \\
e \cdot 1 \equivka e \and
e \equiv 1 \cdot e \and
e \cdot 0 \equivka 0 \and
0 \equiv 0 \cdot e \and
e \cdot (f \cdot g) \equivka (e \cdot f) \cdot g \\
e \cdot (f + g) \equivka e \cdot f + e \cdot g \and
1 + e \cdot e^\star \equivka e^\star \and
e + f \cdot g \leqqka g \implies f^\star \cdot e \leqqka g \phantom{\ ,} \\
(e + f) \cdot g \equivka e \cdot g + f \cdot g \and
1 + e^\star \cdot e \equivka e^\star \and
e + f \cdot g \leqqka f \implies e \cdot g^\star \leqqka f \ ,
\end{mathpar}
in which $e \leqqka f$ is a shorthand for $e + f \equivka f$.

The semantics of rational terms is given by $\semka{-}: \termsrat \to 2^{\Sigma^\star}$, in the following sense:
\begin{align*}
\semka{0} &= \emptyset
    & \semka{a} &= \{ a \}
    & \semka{e + f} &= \semka{e} \cup \semka{f} \\
\semka{1} &= \{ \epsilon \}
    & \semka{e^\star} &= \semka{e}^\star
    & \semka{e \cdot f} &= \semka{e} \cdot \semka{f}
      \ .
\end{align*}
It is furthermore known that $\semka{-}$ characterises $\equivka$~\cite{boffa-1990,krob-1991,kozen-1994}, as follows:
\begin{thrm}[Completeness for \KA]%
\label{theorem:ka-completeness}
Let $e, f \in \termsrat$; now $e \equivka f$ if and only if $\semka{e} = \semka{f}$.
\end{thrm}

We also work with matrices and vectors of rational terms.
Let $Q$ be a finite set.
A \emph{$Q$-vector} is a function $x: Q \to \termsrat$; a \emph{$Q$-matrix} is a function $M: Q \times Q \to \termsrat$.

Let $e \in \termsrat$, let $x$ and $y$ be $Q$-vectors, and let $M$ be a $Q$-matrix.
\emph{Addition} of vectors is defined pointwise, i.e., $x + y$ is the $Q$-vector given by $(x + y)(q) = x(q) + y(q)$.
We can also \emph{scale} vectors, writing $x \fatsemi e$ for the $Q$-vector given by $(x \fatsemi e)(q) = x(q) \cdot e$.

\emph{Multiplication} of a $Q$-vector by a $Q$-matrix yields a $Q$-vector, as expected:
\[
    (M \cdot x)(q) = \sum_{q' \in Q} M(q, q') \cdot x(q') \ .
\]
Here, $\sum$ is the usual generalisation of $+$ for some (arbitrary) choice of bracketing and order on terms; the empty sum is defined to be $0$.

We write $x \equivka y$ when $x$ and $y$ are pointwise equivalent, i.e., for all $q \in Q$ we have $x(q) \equivka y(q)$; we extend $\leqqka$ to $Q$-vectors as before.
Matrices over rational terms again obey the axioms of Kleene algebra~\cite{kozen-1994}.
As a special case, we can obtain the following:
\begin{lem}%
\label{lemma:ka-matrix-fixpoint}
Let $M$ be a $Q$-matrix.
Using the entries of $M$ and applying the operators of Kleene algebra, we can construct a $Q$-matrix $M^\star$, which has the following property.
Let $y$ be any $Q$-vector; now $M^\star \cdot y$ is the least (w.r.t.\ $\leqqka$) $Q$-vector $x$ such that $M \cdot x + y \leqqka x$.
\end{lem}

\subparagraph*{Automata and bisimulations}
We briefly recall \emph{bisimulation up to congruence} for language equivalence of automata, from~\cite{bonchi-pous-2013}.
This will be used in \cref{section:decision-procedure}.

A \emph{non-deterministic automaton} (\emph{NDA}) over an alphabet $\Sigma$ is a triple $(X,o,d)$ where $o \colon X \rightarrow 2$ is an output function, and $d \colon X \times \Sigma \rightarrow X$ a transition function. A non-deterministic finite automaton (NFA) is an NDA where $X$ is finite.
It will be convenient to characterise the semantics of an NDA $(X,o,d)$ recursively as the unique map $\ell: X \to 2^{\Sigma^\star}$ such that
\begin{mathpar}
\ell(x) = \{ \epsilon : o(x) = 1 \} \cup \bigcup_{x' \in d(x, a)} \{ a \} \cdot \ell(x') \ .
\end{mathpar}
This coincides with the standard definition of language acceptance for NDAs.
The \emph{determinisation} of an NDA $(X,o,d)$ is the deterministic automaton $(2^X, \overline{o}, \overline{d})$~\cite{rabin-scott-1959},
where
\begin{mathpar}
\overline{o}(V) =
\begin{cases}
1 & \text{if } \exists s\in V \text{ s.t. } o(s)=1 \\
0 & \text{otherwise }
\end{cases}
\and
\overline{d}(V,a)= \bigcup_{s \in V} d(s,a)
\ .
\end{mathpar}
The \emph{congruence closure} of a relation $R \subseteq 2^X \times 2^X$, denoted $R^c$, is the least equivalence relation such that $R \subseteq R^c$, and if $(C,D)\in R^c$ and $(E,F)\in R^c$ then $(C\cup E, D\cup F )\in R^c$.
A \emph{bisimulation up to congruence} for an NDA  $(X,o,d)$ is a relation $R \subseteq 2^X \times 2^X$ such that for all $(V,W) \in R$, we have $\bar{o}(V)=\bar{o}(W)$ and furthermore for all $a \in \Sigma$, we have $(\bar{d}(V,a),\bar{d}(W,a))\in R^c$.
Bisimulations up to congruence give a proof technique for language equivalence of states of non-deterministic automata, as a consequence of the following~\cite{bonchi-pous-2013}.

\begin{thrm}\label{theorem:bisim}
	Let $(X,o,d)$ be an NDA\@. For all $U, V \in 2^X$,
	we have $\bigcup_{x \in U} \ell(x) = \bigcup_{x \in V} \ell(x)$ if and only if
	there exists a bisimulation up to congruence $R$ for $(X,o,d)$ such that
	$(U,V)\in R$.
\end{thrm}
In~\cite{bonchi-pous-2013}, it is shown how the construction of bisimulations up to congruence leads to a very efficient algorithm for language equivalence.

\section{The problem with CKAT}%
\label{section:problem-with-ckat}

Our motivation for adjusting the axioms of KAT is based on the fact that conjunction and sequential composition are distinct when interleaving is involved, because sequential composition leaves a ``gap'' between tests that might be used by another thread, possibly changing the outcome of the second test.
We now formalise this, by detailing how combining KAT and CKA to obtain CKAT (as in~\cite{jipsen-moshier-2016}) leads to the absurd equation $p \cdot e \cdot \overline{p} \equivckat 0$.

\subparagraph*{Kleene algebra with tests}
To obtain KAT, we enrich rational terms with propositions; concretely, the set of \emph{guarded rational terms} over $\Sigma$ and $\Omega$, denoted $\termsgrat$, is generated by
\[
    e, f ::= 0 \pipe 1 \pipe a \in \Sigma \pipe p \in \termsprop \pipe e + f \pipe e \cdot f \pipe e^\star \ .
\]
We define $\equivkat$ as the smallest congruence on $\termsgrat$ which contains $\equivba$ and obeys the axioms of $\equivka$ (e.g., $e + e \equivkat e$).
Furthermore, $\equivkat$ should relate constants and operators on propositional subterms: for all $p, q \in \termsprop$ it holds that%
\footnote{%
    This is a slightly contrived definition of KAT\@; one usually presents the constants and operators using the same symbols~\cite{kozen-1996}.
    To contrast KAO and KAT it is helpful to make this identification explicit.
}%
\begin{mathpar}
\bot \equivkat 0 \and
\top \equivkat 1 \and
p \vee q \equivkat p + q \and
p \wedge q \equivkat p \cdot q
\ .
\end{mathpar}

Guarded rational terms relate to programs by viewing actions as statements and propositions as assertions.
The last axiom is therefore not strange: if we assert $\mathtt{x = 1}$ followed by $\mathtt{y = 1}$ then surely, if $\mathtt{x}$ and $\mathtt{y}$ remain the same, this is equivalent to asserting $\mathtt{x = 1 \wedge y = 1}$.

\subparagraph*{Concurrent Kleene algebra}
To obtain CKA, we add a parallel composition operator to KA\@.
Concretely, the set of \emph{series-rational terms}~\cite{lodaya-weil-2000} over $\Sigma$, denoted $\termssrat$, is generated by
\[
    e, f ::= 0 \pipe 1 \pipe a \in \Sigma \pipe e + f \pipe e \cdot f \pipe e \parallel f \pipe e^\star \ .
\]
We define $\equivcka$ as the smallest congruence on $\termssrat$ which obeys the axioms that generate $\equivka$, as well as the following for all $e, f, g, h \in \termssrat$:
\begin{mathpar}
e \parallel f \equivcka f \parallel e \and
e \parallel 1 \equivcka e \and
e \parallel 0 \equivcka 0 \and
(e + f) \parallel g \equivcka e \parallel g + f \parallel g \\
e \parallel (f \parallel g) \equivcka (e \parallel f) \parallel g \and
(e \parallel f) \cdot (g \parallel h) \leqqcka (e \cdot g) \parallel (f \cdot h) \ ,
\end{mathpar}
in which $e \leqqcka f$ is an abbreviation for $e + f \equivcka f$.

Here, $0$ annihilates parallel composition because $0$ corresponds to the program without valid traces; hence, running $e$ in parallel with $0$ also cannot yield any traces.
Distributivity of $\parallel$ over $+$ witnesses that a choice within a thread can also be made outside that thread.

The last axiom, called the \emph{exchange law}~\cite{hoare-moeller-struth-wehrman-2009}, encodes interleaving.
Intuitively, it says that when programs $e \cdot g$ and $f \cdot h$ run in parallel, their behaviour includes running their heads in parallel (i.e., $e \parallel f$) followed by their tails (i.e., $g \parallel h$).
Taken to its extreme, the exchange law says that the behaviour of a program includes its complete linearisations.

\subparagraph*{Concurrent Kleene algebra with tests}
To define CKAT~\cite{jipsen-moshier-2016}, we choose \emph{guarded series-rational terms} over $\Sigma$ and $\Omega$, denoted $\termsgsrat$, as those generated by the grammar
\[
    e, f ::= 0 \pipe 1 \pipe a \in \Sigma \pipe p \in \termsprop \pipe e + f \pipe e \cdot f \pipe e \parallel f \pipe e^\star \ .
\]
Now, $\equivckat$ is the least congruence on $\termsgsrat$ obeying the axioms of $\equivkat$ and $\equivcka$.

If we view guarded series-rational terms as representations of programs, and use $\equivckat$ to reason about them, then we should expect to obtain sensible statements about programs, since the axioms that underpin CKAT seem reasonable in that context.

To test this hypothesis, consider the guarded series-rational term $p \cdot e \cdot \overline{p}$, which represents a program that first performs an assertion $p$, then runs a program described by $e$, and asserts the negation of $p$.
There are choices of $p$ and $e$ that should make $p \cdot e \cdot \overline{p}$ describe a program with valid behaviour; for instance, $p$ could assert that $\mathtt{x = 1}$, and $e$ could be the program $\mathtt{x \leftarrow 0}$.
Hence, by our hypothesis, $p \cdot e \cdot \overline{p}$ should not, in general, be equivalent to $0$, the program without any valid behaviour.
Unfortunately, the opposite is true.

\begin{ant}%
\label{antinomy:problem}
Let $e \in \termsgsrat$ and $p \in \termsprop$; now $p \cdot e \cdot \overline{p} \equivckat 0$.
\end{ant}
\begin{proof}
By the axiom that $1$ is the unit of $\parallel$ and the exchange law, we derive that
\[
    p \cdot e \cdot \overline{p}
        \equivckat (p \parallel 1) \cdot (1 \parallel e) \cdot \overline{p}
        \leqqckat ((p \cdot 1) \parallel (1 \cdot e)) \cdot \overline{p}
        \ .
\]
By the same reasoning, and the axiom that $1$ is the unit of $\cdot$, we have that
\[
    ((p \cdot 1) \parallel (1 \cdot e)) \cdot \overline{p}
        \equivckat (p \parallel e) \cdot (\overline{p} \parallel 1)
        \leqqckat (p \cdot \overline{p}) \parallel (e \cdot 1)
        \ .
\]
Applying the unit, and using the identification of $\wedge$ and $\cdot$ on tests, we find that
\[
    (p \cdot \overline{p}) \parallel (e \cdot 1)
        \equivckat (p \cdot \overline{p}) \parallel e
        \equivckat (p \wedge \overline{p}) \parallel e
        \equivckat \bot \parallel e
        \ .
\]
Since $\bot$ and $0$ are identified, and $0$ annihilates parallel composition, we have that
\[
    \bot \parallel e
        \equivckat 0 \parallel e
        \equivckat 0
        \ .
\]
By the above, we have $p \cdot e \cdot \overline{p} \leqqckat 0$.
Since $0 \leqqckat p \cdot e \cdot \overline{p}$, the claim follows.
\end{proof}

Another way of contextualising the above is to consider that propositional Hoare logic can be encoded in KAT~\cite{kozen-2001}.
Concretely, a propositional Hoare triple $\{ p \} S \{ q \}$ is valid if $p \cdot e_S \cdot \overline{q} \equivkat 0$, where $e_S$ is a straightforward encoding of $S$.
It stands to reason that sequential programs should similarly encode in CKAT\@.
However, if we apply that line of reasoning to the above, then \emph{any} Hoare triple of the form $\{ p \} S \{ p \}$ is valid, which would mean that \emph{any property is an invariant of any program}.

\section{Kleene algebra with observations}%
\label{section:introducing-kao}

We now propose KAO as an alternative way of embedding Boolean guards in rational terms.
This approach should prevent \cref{antinomy:problem}, and thus make KAO more suitable for an extension with concurrency.
We start by motivating how we adapt KAT, before proceeding with a language model and a generalisation of partial derivatives to KAO\@.

\subsection{From tests to observations}
The root of the problem in \cref{antinomy:problem} is the axiom $p \cdot q \equivkat p \wedge q$, telling us that $(p \cdot \overline{p}) \parallel e \equivckat (p \wedge \overline{p}) \parallel e$.
Interpreted as programs, these are different: in one, $e$ can be interleaved between $p$ and $\overline{p}$, which the other does not allow.

To fix this problem, we propose a new perspective on Boolean guards.
Rather than considering a guard a \emph{test}, which in KAT entails an assertion valid from the last action to the next, we consider a guard an \emph{observation}, i.e., an assertion valid at that particular point in the execution of the program.
This weakens the connection between conjunction and concatenation: if a program observes $p$ followed by $q$, there is no guarantee that $p$ and $q$ can be observed simultaneously.
Hence, we drop the equivalence between conjunction and concatenation of tests.
We call this system \emph{weak Kleene algebra with observations}; like KAT, its axioms are based on the axioms of Boolean algebra and Kleene algebra, as follows.
\begin{definition}[\WKAO axioms]
We define $\equivwkao$ as the smallest congruence on $\termsgrat$ that contains $\equivba$ and also obeys the axioms of $\equivka$.
Furthermore, $\bot \equivwkao 0$ and for all $p, q \in \termsprop$ it holds that $p \vee q \equivwkao p + q$.
\end{definition}

Since conjunction and concatenation no longer coincide, we have also dropped the axiom that relates their units.
This can be justified from our shift in perspective: $\top$, i.e., the observation that always succeeds, is an action that leaves some record in the behaviour of the program, whereas $1$, i.e., the program that does nothing, has no such obligation.%
\footnote{%
    We refer to~\cite[Appendix E]{kao-full} for further algebraic details on the consequences of identifying these units.
}

As hinted before, it may happen that when a program observes $p$ followed by $q$, both these assertions are true simultaneously, i.e., their conjunction could have been observed; hence, the behaviour of observing $p$ \emph{and} $q$ should be contained in the behaviour of observing $p$ \emph{and then} $q$.
For instance, consider the pseudo-programs $S = \mathtt{assert}\ p \wedge q;\ e$ and $S' = \mathtt{assert}\ p;\ \mathtt{assert}\ q;\ e$.
Here, $S$ asserts that $p$ and $q$ hold at the same time before proceeding with $e$, while $S'$ first asserts $p$, and then $q$.
The behaviour of $S$ should be included in that of $S'$: if $p$ and $q$ are simultaneously observable, then the first two observations might take place simultaneously.
Encoding this in an additional axiom leads to KAO proper.

\begin{definition}[\KAO axioms]
We define $\equivkao$ as the smallest congruence on $\termsgrat$ that contains $\equivwkao$, and furthermore satisfies the \emph{contraction law}: for all $p, q \in \termsprop$ it holds that $p \wedge q \leqqkao p \cdot q$ --- where $e \leqqkao f$ is a shorthand for $e + f \equivkao f$.
\end{definition}

We briefly return to our example programs.
If we encode $S$ as $(p \wedge q) \cdot e$ and $S'$ as $p \cdot q \cdot e$, then the behaviour of $S$ is contained in that of $S'$, because $(p \wedge q) \cdot e \leqqkao p \cdot q \cdot e$.

\begin{remark}
  In the presence of the other axioms, $p \wedge q \leqqkao p \cdot q$ is equivalent to $p\leqqkao p\cdot p$.
  Indeed, the second inequality may be inferred from the first, since $p\equivkao p\wedge p$.
  The converse implication is obtained as follows: $p\wedge q\leqqkao \left(p\wedge q\right)\cdot\left(p\wedge q\right)\leqqkao p\cdot q$, using that $p\wedge q\leqqkao p,q$.
  While the contraction law is more convenient to compare terms, the axiom $p\leqqkao p\cdot p$ will serve implicitly as the basis for the contraction relation $\contr$ defined in the next section.
\end{remark}

\subsection{A language model}

The Kleene algebra variants encountered thus far have models based on rational languages, which characterise the equivalences derivable using the axioms.
To find a model for guarded rational terms which corresponds to KAO, we start with a model of weak KAO\@:
\begin{definition}[\WKAO semantics]
Let $\Gamma = \at \cup \Sigma$.
Languages over $\Gamma$ are called \emph{observation languages}.
We define $\semwkao{-}: \termsgrat \to 2^{\Gamma^\star}$ as follows:
\begin{align*}
\semwkao{0} &= \emptyset       & \semwkao{a} &= \{ a \}    & \semwkao{e + f}     &= \semwkao{e} \cup \semwkao{f}  & \semwkao{e^\star}       &= \semwkao{e}^\star \\
\semwkao{1} &= \{ \epsilon \}  & \semwkao{p} &= \semba{p}  & \semwkao{e \cdot f} &= \semwkao{e} \cdot \semwkao{f} \ ,
\end{align*}
where it is understood that $\semba{p} \subseteq \at \subseteq \Gamma \subseteq \Gamma^\star$.
\end{definition}

Observation languages are a model for guarded rational terms w.r.t.\ $\equivwkao$:

\begin{restatable}[\WKAO soundness]{lem}{wkaosoundness}%
\label{lemma:wkao-soundness}
Let $e, f \in \termsgrat$.
Now, if $e \equivwkao f$, then $\semwkao{e} = \semwkao{f}$.
\end{restatable}

More work is necessary to get a model of guarded rational terms w.r.t.\ $\equivkao$.
In particular, $\semwkao{-}$ does not preserve the contraction law: if $o_1, o_2 \in \Omega$, then
\begin{mathpar}
\semwkao{o_1 \wedge o_2} = \{ \alpha \in \at : o_1, o_2 \in \alpha \} \and
\semwkao{o_1 \cdot o_2}  = \{ \alpha\beta \in \at \cdot \at : o_1 \in \alpha, o_2 \in \beta \} \ ,
\end{mathpar}
meaning that $\semwkao{o_1 \wedge o_2} \not\subseteq \semwkao{o_1 \cdot o_2}$, despite $o_1 \wedge o_2 \leqqkao o_1 \cdot o_2$.%
\footnote{%
    Incidentally, this shows that the contraction law is independent of the axioms that build $\equivwkao$, i.e., that those axioms are not sufficient to prove the contraction law.
}

To obtain a sound model, we need the counterpart of the contraction law on the level of observation languages, which works out as follows

\begin{definition}[Contraction]
We define $\contr$ as the smallest relation on $\Gamma^\star$ s.t.
\begin{mathpar}
\inferrule{%
    w,x \in \Gamma^\star \\
    \alpha \in \at
}{%
    w\alpha{}x \contr w\alpha\alpha{}x
}
\ .
\end{mathpar}
When $L \subseteq \Gamma^\star$, we write $\down{L}$ for the $\contr$-closure of $L$, that is to say, the smallest observation language such that $L \subseteq \down{L}$, and if $w \contr x$ with $x \in \down{L}$, then $w \in \down{L}$.
\end{definition}
Thus, $\down{L}$ contains all words obtained from words in $L$ by contracting any number of repeated atoms (but not actions) into one.
Applying this closure to the semantics encodes the intuition that repeated observations may also correspond to just one step in the execution.

With these tools, we can define a model of KAO as follows.
\begin{definition}[\KAO semantics]
We define $\semkao{-}: \termsgrat \to 2^{\Gamma^\star}$ by $\semkao{e} = \down{\semwkao{e}}$.
\end{definition}

Alternatively, we can describe $\semkao{-}$ compositionally, using the following.

\begin{restatable}{lem}{closurealmostcommutes}%
\label{lemma:closure-almost-commutes}
Let $e, f \in \termsgrat$.
The following hold:
    (i)~$\semkao{e + f} = \semkao{e} \cup \semkao{f}$, and
    (ii)~$\semkao{e \cdot f} = \down{(\semkao{e} \cdot \semkao{f})}$, and
    (iii)~$\semkao{e^\star} = \down{\semkao{e}^\star}$.
\end{restatable}

It is now straightforward to check that $\semkao{-}$ preserves the axioms of $\equivkao$.
\begin{restatable}[\KAO soundness]{lem}{soundness}%
\label{lemma:soundness}
Let $e, f \in \termsgrat$.
Now, if $e \equivkao f$, then $\semkao{e} = \semkao{f}$.
\end{restatable}

\subsection{Partial derivatives}\label{partial-derivatives}

\emph{Derivatives}~\cite{brzozowski-1964} are a powerful tool in Kleene algebra variants.
In the context of programs, we can think of derivatives as an operational semantics; they tell us whether the term represents a program that can halt immediately (given by the \emph{termination map}), as well as the program that remains to be executed once an action is performed (given by the \emph{continuation map}).

Derivatives are typically compatible with the congruences used to reason about terms; what's more, they are closely connected to finite automata~\cite{brzozowski-1964}.
This makes them useful for reasoning about language models~\cite{kozen-2001,bonchi-pous-2013,pous-2014}.
We therefore introduce a form of derivatives of guarded rational terms that works well with $\equivkao$.
Let us start by giving the termination map, which is close to the termination map of rational terms compatible with $\equivka$.
\begin{definition}[Termination]\label{def:epsilon}
We define $\epsilon: \termsgrat \to 2$ inductively, as follows:
\begin{align*}
\epsilon(0) &= 0  & \epsilon(a) &= 0  & \epsilon(e + f)     &= \max \{ \epsilon(e), \epsilon(f) \}  & \epsilon(e^\star) &= 1 \\
\epsilon(1) &= 1  & \epsilon(p) &= 0  & \epsilon(e \cdot f) &= \min \{ \epsilon(e), \epsilon(f) \} \ .
\end{align*}
\end{definition}

To check that $\epsilon$ characterises guarded rational terms that can terminate, we record the following lemma, which says that $\epsilon(e) = 1$ precisely when $\semkao{e}$ includes the empty string.

\begin{restatable}{lem}{accvsunit}%
\label{lemma:acc-vs-unit}
Let $e \in \termsgrat$.
Now $\epsilon(e) \leqqkao e$, and $\epsilon(e) = 1$ if and only if $\epsilon \in \semkao{e}$.
\end{restatable}

Next we define the continuation map, or more specifically, \emph{partial continuation map}~\cite{antimirov-1996}: given $\mathfrak{a} \in \Gamma$, this function gives a \emph{set} of terms representing possible continuations of the program after performing $\mathfrak{a}$.
We start with the continuation map for actions.

\begin{definition}[$\Sigma$-continuation]
We define $\delta: \termsgrat \times \Sigma \to 2^{\termsgrat}$ inductively
\begin{align*}
\delta(0, a)    &= \delta(1, a) = \emptyset & \delta(e + f, a)      &= \delta(e, a) \cup \delta(f, a)                               \\
\delta(a, a')   &= \{ 1 : a = a' \}         & \delta(e \cdot f, a)  &= \{ e' \cdot f : e' \in \delta(e, a) \} \cup \Delta(e, f, a)  \\
\delta(p, a)    &= \emptyset                & \delta(e^\star, a)    &= \{ e' \cdot e^\star : e' \in \delta(e, a) \} \ ,
\end{align*}
where $\Delta(e, f, a) = \delta(f, a)$ when $\epsilon(e) = 1$, and $\emptyset$ otherwise.
\end{definition}

Finally, we give the continuation map for observations, which is similar to the one for actions, except that on sequential composition it is subtly different.

\begin{definition}[$\at$-continuation]
We define $\zeta: \termsgrat \times \at \to 2^{\termsgrat}$ inductively
\begin{align*}
\zeta(0, \alpha) &= \delta(1, \alpha) = \emptyset    & \zeta(e + f, \alpha) &= \zeta(e, \alpha) \cup \zeta(f, \alpha)                               \\
\zeta(a, \alpha) &= \emptyset                        & \zeta(e \cdot f, \alpha) &= \{ e' \cdot f : e' \in \zeta(e, \alpha) \} \cup Z(e, f, \alpha)  \\
\zeta(p, \alpha) &= \{ 1 : \pi_\alpha \leqqba p \}   & \zeta(e^\star, \alpha) &= \{ e' \cdot e^\star : e' \in \zeta(e, \alpha) \} \ ,
\end{align*}
where $Z(e, f, \alpha) = \zeta(f, \alpha)$ when $\epsilon(e) = 1$ or $\epsilon(e') = 1$ for some $e' \in \zeta(e, \alpha)$, and $\emptyset$ otherwise.
\end{definition}

For instance, let $p, q \in \termsprop$.
If $\alpha \in \semba{p} \cap \semba{q}$, then we can calculate that
\[
    \zeta(p \cdot q, \alpha)
        = \{ p' \cdot q : p' \in \zeta(p, \alpha) \} \cup Z(p, q, \alpha)
        = \{ 1 \cdot q \} \cup \{ 1 \} \ ,
\]
which is to say that the program can either continue in $1 \cdot q$, where it has to make an observation validating $q$, or it can choose to re-use the observation $\alpha$ to validate $q$, continuing in $Z(p, q, \alpha) = \zeta(q, \alpha) = \{ 1 \}$, because $1 \in \zeta(p, \alpha)$.
This notion that $\zeta$ can apply an observation more than once to validate multiple assertions in a row can be formalised.

\begin{restatable}{lem}{doublederivatives}%
\label{lemma:double-derivatives}
Let $e \in \termsgrat$, $\alpha \in \at$, and $e' \in \zeta(e, \alpha)$.
Now $\zeta(e', \alpha) \subseteq \zeta(e, \alpha)$.
\end{restatable}

To check that $\delta$ and $\zeta$ indeed output continuations of the term after the action or assertion has been performed, we should check that they are compatible with $\equivkao$.

\begin{restatable}{lem}{halffundamentaltheorem}%
\label{lemma:half-fundamental-theorem}
Let $e \in \termsgrat$.
For all $a \in \Sigma$ and $e' \in \delta(e, a)$, we have $a \cdot e' \leqqkao e$.
Furthermore, for all $\alpha \in \at$ and $e' \in \zeta(e, \alpha)$, we have $\pi_\alpha \cdot e' \leqqkao e$.
\end{restatable}

We also need the \emph{reach} of a guarded rational term, which is meant to describe the set of terms that can be obtained by repeatedly applying continuation maps.

\begin{definition}[Reach of a term]
For $e \in \termsgrat$, we define $\rho: \termsgrat \to 2^{\termsgrat}$ inductively:
\begin{align*}
\rho(0) &= \emptyset  & \rho(a) &= \{ 1, a \}  & \rho(e + f) &= \rho(e) \cup \rho(f)                               \\
\rho(1) &= \{ 1 \}    & \rho(p) &= \{ 1, p \}  & \rho(e \cdot f) &= \{ e' \cdot f : e' \in \rho(e) \} \cup \rho(f) \\
        &             &         &              & \rho(e^\star) &= \{ 1 \} \cup \{ e' \cdot e^\star : e' \in \rho(e) \} \ .
\end{align*}
\end{definition}

Indeed, $\rho(e)$ truly contains all terms reachable from $e$ by means of $\delta$ and $\zeta$:

\begin{restatable}{lem}{reachcontainsderivatives}%
\label{lemma:reach-contains-derivatives}
Let $e \in \termsgrat$.
If $a \in \Sigma$, then $\delta(e, a) \subseteq \rho(e)$; if $e' \in \rho(e)$, then $\delta(e', a) \subseteq \rho(e)$.
Furthermore, if $\alpha \in \at$, then $\zeta(e, \alpha) \subseteq \rho(e)$; if $e' \in \rho(e)$, then $\zeta(e', \alpha) \subseteq \rho(e)$.
\end{restatable}

It is not hard to see that for all $e \in \termsgrat$, we have that $\rho(e)$ is finite.
Note that $\rho(e)$ does not, in general, contain $e$ itself.
On the other hand, $\rho(e)$ does contain a set of terms that is sufficient to reconstruct $e$, up to $\equivkao$.
We describe these as follows.

\begin{definition}[Initial factors]\label{def:iota}
For $e \in \termsgrat$, we define $\iota: \termsgrat \to 2^{\termsgrat}$ inductively:
\begin{align*}
\iota(0) &= \emptyset  & \iota(a) &= \{ a \}  & \iota(e + f) &= \iota(e) \cup \iota(f)                            \\
\iota(1) &= \{ 1 \}    & \iota(p) &= \{ p \}  & \iota(e \cdot f) &= \{ e' \cdot f : e' \in \iota(e) \} \\
         &             &          &           & \iota(e^\star) &= \{ 1 \} \cup \{ e' \cdot e^\star : e' \in \iota(e) \} \ .
\end{align*}
\end{definition}

Now, to reconstruct $e$ from $\iota(e)$, all we have to do is sum its elements.
\begin{restatable}{lem}{reconstructfrominitial}%
\label{lemma:reconstruct-from-initial}
Let $e \in \termsgrat$.
Then $e \equivkao \sum_{e' \in \iota(e)} e'$.
\end{restatable}

\section{Completeness}%
\label{section:completeness}

We now set out to show that $\semkao{-}$ characterises $\equivkao$.
Since soundness of $\semkao{-}$ w.r.t.\ $\equivkao$ was already shown in \cref{lemma:soundness}, it remains to prove completeness, i.e., if $e$ and $f$ are interpreted equally by $\semkao{-}$, then they can be proven equivalent via $\equivkao$.
To this end, we first identify a subset of guarded rational terms for which a completeness result can be shown by relying on the completeness result for KA\@.
To describe these terms, we need the following.%

\begin{definition}[Atomic and guarded terms]
Let $\Pi$ denote the set $\{ \pi_\alpha : \alpha \in \at \}$.
The set of \emph{atomic guarded rational terms}, denoted $\termsagrat$, is the subset of $\termsgrat$ generated by the grammar
\[
    e, f ::= 0 \pipe 1 \pipe a \in \Sigma \pipe \pi_\alpha \in \Pi \pipe e + f \pipe e \cdot f \pipe e^\star \ .
\]
Furthermore, if $e \in \termsgrat$ such that $\semkao{e} = \semwkao{e}$, we say that $e$ is \emph{closed}.
\end{definition}

Completeness of $\equivkao$ with respect to $\semkao{-}$ can then be established for atomic and closed guarded rational terms as follows.

\begin{restatable}{prop}{kaocompletenesspartial}%
\label{proposition:kao-completeness-partial}
Let $e, f \in \termsagrat$ be closed, and $\semkao{e} = \semkao{f}$; then $e \equivkao f$.
\end{restatable}
\begin{proof}[Sketch]
By noting that the semantics of an atomic and closed term coincide with the KA-semantics (over an extended alphabet); the result then follows by appealing to completeness of KA, and the fact that the axioms of KA are contained in those for KAO\@.
\end{proof}

To show that $\semkao{-}$ characterises $\equivkao$ for \emph{all} guarded rational terms, it suffices to find for every $e \in \termsgrat$ a closed $\hat{e} \in \termsagrat$ with $e \equivkao \hat{e}$.
After all, if we have such a transformation, then $\semkao{e} = \semkao{f}$ implies $\semkao{\hat{e}} = \semkao{\hat{f}}$, which implies $\hat{e} \equivkao \hat{f}$, and hence $e \equivkao f$.

In the remainder of this section, we describe how to obtain a closed and atomic guarded rational term from a guarded rational term $e$.
This transformation works by first ``disassembling'' $e$ using partial derivatives; on an intuitive level, this is akin to creating a (non-deterministic finite) automaton that accepts the observation language described by $e$.
Next, we ``reassemble'' from this representation an atomic term $\hat{e}$, equivalent to $e$; this is analogous to constructing a rational expression from the automaton.
To show that $\hat{e}$ is closed, we leverage \cref{lemma:double-derivatives}, essentially arguing that the automaton obtained from $e$ encodes $\contr$.

We extend the notation for vectors and matrices over rational terms to guarded rational terms.
We can then represent a guarded rational term by a matrix and a vector, as follows.
\begin{definition}
For $e \in \termsgrat$, define the $\rho(e)$-vector $x_e$ and $\rho(e)$-matrix $M_e$ by
\begin{mathpar}
x_e(e') = \epsilon(e')
\and
M_e(e', e'') = \sum_{e'' \in \delta(e', a)} a + \sum_{e'' \in \zeta(e', \alpha)} \pi_\alpha \ .
\end{mathpar}
\end{definition}
Note that, in the above, $\delta(e', a)$ and $\zeta(e', \alpha)$ are finite by \cref{lemma:reach-contains-derivatives}.

Kozen's argument that matrices over rational terms satisfy the axioms of Kleene algebra~\cite{kozen-1994} generalises to guarded rational terms.
This leads to a straightforward generalisation of \cref{lemma:ka-matrix-fixpoint}.
For the sake of brevity, we use $\AKA$ to stand in for $\WKAO$ or $\KAO$.

\begin{lem}%
\label{lemma:matrix-fixpoint-kao}
Let $M$ be a $Q$-matrix.
Using the entries of $M$ and applying the operators of Kleene algebra, we can construct a $Q$-matrix $M^\star$, which has the following property.
Let $y$ be any $Q$-vector; now $M^\star \cdot y$ is the least (w.r.t.\ $\leqqaka$) $Q$-vector $x$ such that $M \cdot x + y \leqqaka x$.
\end{lem}

We can now use the above to reassemble a term from $M_e$ and $x_e$ as follows.
\begin{definition}[Transformation of terms]
Let $e \in \termsgrat$.
We write $s_e$ for the $\rho(e)$-vector given by $M_e^\star \cdot x_e$, and $\hat{e}$ for the guarded rational term given by $\sum_{e' \in \iota(e)} s_e(e')$.
\end{definition}

Since $\hat{e}$ is constructed from $M_e$ and $s_e$, and these are built using atomic guarded rational terms, $\hat{e}$ is atomic.
We carry on to show that $\hat{e} \equivkao e$.
To this end, the following is useful.

\begin{restatable}[Least solutions]{prop}{leastsolutioncharacterisation}%
\label{proposition:least-solution-characterisation}
Let $e, f \in \termsgrat$.
Now $s_e \fatsemi f$ is the least (w.r.t. $\leqqaka$) $\rho(e)$-vector $s$ such that for each $e' \in \rho(e)$, it holds that
\[
    \epsilon(e') \cdot f
        + \sum_{e'' \in \delta(e', a)} a \cdot s(e'')
        + \sum_{e'' \in \zeta(e', \alpha)} \pi_\alpha \cdot s(e'')
            \leqqaka s(e')
            \ .
\]
When $s = s_e \fatsemi f$, the above is an equivalence, i.e., we can substitute $\leqqaka$ with $\equivaka$.
\end{restatable}
\begin{proof}[Sketch]
The least such $s$ coincides with $s_e = M_e^\star \cdot x_e$, by using the property of $s_e$ that we obtain as a result of \cref{lemma:matrix-fixpoint-kao}.
The latter claim goes by standard argument akin to the argument showing that the least pre-fixpoint of a monotone operator is a fixpoint.
\end{proof}

Using the above, we show that the contents of $\rho(e)$ themselves qualify as a least $\rho(e)$-vector satisfying system obtained from $e$ (and fixing $f = 1$).
More concretely, we have the following.

\begin{restatable}{prop}{leastsolutionidentity}%
\label{proposition:least-solution-identity}
Let $e \in \termsgrat$ and $e' \in \rho(e)$.
Then $s_e(e') \equivkao e'$.
\end{restatable}
\begin{proof}[Sketch]
The key idea is to first relate the least $\rho(e)$-vectors satisfying the system obtained from $e$ to those satisfying the systems arising from its subterms.
The proof then follows by induction on $e$, and some straightforward derivations.
\end{proof}

This allows us to conclude (using \cref{lemma:reconstruct-from-initial}) that $e \equivkao \hat{e}$.

\begin{restatable}[Transformation preserves equivalence]{cor}{leastsolutionidentitytotal}%
\label{corollary:least-solution-identity-total}
Let $e \in \termsgrat$; then $e \equivkao \hat{e}$.
\end{restatable}

It remains to show that $\hat{e}$ is closed.
To this end, the following characterisation is helpful:

\begin{restatable}{lem}{closureequivalence}%
\label{lemma:closure-equivalence}
Let $L \subseteq \Gamma^\star$, and define $\altcontr$ as the smallest relation on $\Gamma^\star$ satisfying the rules
\begin{mathpar}
\inferrule{~}{%
    \epsilon \altcontr \epsilon
}
\and
\inferrule{%
    \mathfrak{a} \in \Gamma \\
    w \altcontr x
}{%
    \mathfrak{a}w \altcontr \mathfrak{a}x
}
\and
\inferrule{%
    \alpha \in \at \\
    w \altcontr x
}{%
    \alpha{}w \altcontr \alpha\alpha{}x
}
\ .
\end{mathpar}
$\down{L}$ is the smallest subset of $\Gamma^\star$ such that $L \subseteq \down{L}$, and if $w \altcontr x$ with $x \in \down{L}$, then $w \in \down{L}$.
\end{restatable}

We can then employ the characterisation of $s_e$ in \cref{proposition:least-solution-characterisation} to show, by induction on the length of the $\altcontr$-chain, that the components of $s_e$ are closed.

\begin{restatable}{prop}{componentclosure}%
\label{proposition:component-closure}
Let $e \in \terms$.
For all $e' \in \rho(e)$, it holds that $s_e(e')$ is closed.
\end{restatable}
\begin{proof}[Sketch]
Note that $s_e$ satisfies the system obtained from $e$ as in \cref{proposition:least-solution-characterisation}, both w.r.t. $\equivkao$ and $\equivwkao$.
Using \cref{lemma:closure-equivalence}, we can then show (by induction on the number of applications of $\altcontr$) that $\semwkao{s(e')} \subseteq \semkao{s(e')}$; the other inclusion holds trivially.
\end{proof}

Hence, we can conclude (using \cref{lemma:closure-almost-commutes}) that $\hat{e}$ is closed as well, as follows.

\begin{restatable}[Transformation yields closed term]{cor}{closure}%
\label{corollary:closure}
Let $e \in \termsgrat$; now $\hat{e}$ is closed.
\end{restatable}

We are now ready conclude with the main result of this section.

\begin{thrm}[Soundness and completeness]%
\label{theorem:soundness-completeness}
Let $e, f \in \termsgrat$; then
\[
    e \equivkao f
    \iff
    \semkao{e} = \semkao{f}
    \ .
\]
\end{thrm}
\begin{proof}[Sketch]
Since we have a way to obtain from $e \in \termsgrat$ a closed term $\hat{e} \in \termsagrat$, for which it furthermore holds that $e \equivkao \hat{e}$, we can appeal to \cref{proposition:kao-completeness-partial}.
\end{proof}

\section{Decision procedure}%
\label{section:decision-procedure}

We now design a procedure that takes terms $e,f \in \termsgrat$ and decides whether $\semkao{e} = \semkao{f}$; by \cref{theorem:soundness-completeness}, this gives us a decision procedure for $e \equivkao f$.
To this end, we reduce to the problem of language equivalence between NFAs over guarded rational terms, defined using partial derivatives.
This, in turn, can be efficiently decided using bisimulation techniques.

First, observe that $\termsgrat$ carries a non-deterministic automaton structure.
\begin{definition}[Syntactic automaton]\label{def:aut-expr}
The set $\termsgrat$ carries an NDA structure $(\termsgrat, \epsilon, \theta)$ where $\epsilon: \termsgrat \to 2$ is as in \cref{def:epsilon}, and $\theta: \termsgrat \times \Gamma \to 2^{\termsgrat}$ is given by
\[
    \theta(e,\mathfrak{a}) =
    \begin{cases}
    \delta(e, \mathfrak{a}) & \text{if } \mathfrak{a} \in \Sigma \\
    \zeta(e, \mathfrak{a}) & \text{if } \mathfrak{a} \in \at
    \end{cases}
    \ .
\]
We call this the \emph{syntactic automaton} of guarded rational terms.
\end{definition}
Let
 $\ell \colon \termsgrat \rightarrow 2^{\Gamma^{\star}}$ be the semantics of this automaton as given in \cref{section:preliminaries}.
 It is easy to see that $\ell$ is the unique function such that
\begin{equation}\label{language-kao-automaton}
\ell(e)= \{ \epsilon : \epsilon(e) = 1 \} \cup \bigcup_{e'\in\delta(e,a)} \{a\}\cdot \ell(e') \cup \bigcup_{e'\in\zeta(e,\alpha)} \{\alpha\}\cdot \ell(e') \ .
\end{equation}
To use this NDA for \KAO equivalence, we need to show that the language accepted by the state $e \in \termsgrat$ is $\semkao{e}$.
For this goal, it helps to algebraically characterise expressions in terms of their derivatives.
We call this a \emph{fundamental theorem} of \KAO, after~\cite{rutten-2003,silva-2010}.
\begin{restatable}[Fundamental theorem]{thrm}{fundamental}%
\label{theorem:fundamental}
For all $e \in \termsgrat$, the following holds
\[
    e \equivkao \epsilon(e) + \sum_{e' \in \delta(e, a)} a \cdot e' + \sum_{e' \in \zeta(e, \alpha)} \pi_\alpha \cdot e' \ .
\]
\end{restatable}
\begin{proof}[Sketch]
The right-to-left containment is a result of \cref{lemma:acc-vs-unit,lemma:half-fundamental-theorem}.
The converse is a straightforward calculation using \cref{proposition:least-solution-characterisation,corollary:least-solution-identity-total}.
\end{proof}

Next, we turn the fundamental theorem into a statement about the \KAO semantics.
Before we do, however, we need the following basic result about the \KAO semantics.
\begin{restatable}{lem}{nomorearrow}\label{nomorearrow}
Let us fix $\alpha \in \at$.
For all $e \in \termsgrat$, we have:
\[
    \bigsemkao{\sum\nolimits_{e' \in \zeta(e, \alpha)} \pi_\alpha \cdot e'}
        = \bigcup_{e' \in \zeta(e, \alpha)} \{ \alpha \} \cdot \semkao{e'} \ .
\]
\end{restatable}

We now arrive at a semantic analogue of the fundamental theorem.
\begin{restatable}{prop}{fundamentalsemantics}%
\label{lemma:fundamentalsemantics}
For all $e\in\termsgrat$, we have:
\[
    \semkao{e}
        = \{ \epsilon : \epsilon(e) = 1 \}
            \cup \bigcup_{e' \in \delta(e, a)} \{a\} \cdot \semkao{e'}
            \cup \bigcup_{e' \in \zeta(e, \alpha)} \{ \alpha \} \cdot \semkao{e'}
            \ .
\]
\end{restatable}
\begin{proof}[Sketch]
By applying soundness of $\equivkao$ w.r.t. $\semkao{-}$ (\cref{lemma:soundness}) to the fundamental theorem, and using \cref{nomorearrow} for the term summing over the $\at$-continuations of $e$.
\end{proof}

The language of a state in the syntactic NFA is then characterised by the \KAO semantics.

\begin{thrm}[Soundness of translation]\label{theorem:language-equiv}
  Let $e\in\termsgrat$; then $\semkao{e}=\ell(e)$.
\end{thrm}
\begin{proof}[Sketch]
A direct consequence of \cref{lemma:fundamentalsemantics} and the uniqueness of $\ell$ in satisfying~\eqref{language-kao-automaton}.
\end{proof}

To decide whether $\semkao{e}=\semkao{f}$ it suffices to show that $e$ and $f$ are language equivalent in the syntactic automaton.
We express this in terms of $\iota$ as defined in \cref{def:iota}.

\begin{restatable}{cor}{semiota}%
\label{cor:semiota}
	For all $e,f \in \termsgrat$, we have
	\[
        \semkao{e} = \semkao{f} \quad \Leftrightarrow \quad
        \bigcup_{e' \in \iota(e)} \ell(e') = \bigcup_{f' \in \iota(f)} \ell(f')
        \ .
    \]
\end{restatable}

By construction, $\iota(e) \subseteq \rho(e)$ for every $e \in \termsgrat$.
Since $\rho(e)$ and $\rho(f)$ are closed under partial derivatives, we can restrict the syntactic automaton to $\rho(e) \cup \rho(f)$, to obtain a \emph{finite} NDA\@.
To decide $e\equivkao f$, it suffices to decide $\bigcup_{e' \in \iota(e)} \ell(e') = \bigcup_{f' \in \iota(f)} \ell(f')$ on this NFA\@.

This leads us to the main result of this section: a decision procedure for \KAO. 

\begin{restatable}[Decision procedure]{thrm}{decisionprocedure}
For all $e, f \in \termsgrat$, we have $e \equivkao f$ if and only if there exists $R$ such that $(\iota(e),\iota(f)) \in R$ and $R$ is a bisimulation up to congruence for the syntactic automaton restricted to $\rho(e) \cup \rho(f)$.
\end{restatable}
\begin{proof}[Sketch]
By \cref{theorem:bisim}, such an $R$ exists precisely when $\bigcup_{e' \in \iota(e)} \semkao{e'} = \bigcup_{f' \in \iota(f)} \semkao{f'}$; by \cref{cor:semiota,theorem:soundness-completeness} this is equivalent to $e \equivkao f$.
\end{proof}

\begin{example}
We know that for all $a,b\in\Omega$ we have that $a\wedge b \not\equivkao a\cdot b$.
This also follows from our decision procedure, which we argue by showing that any attempt to construct a bisimulation up to congruence $R$ containing $(\{a\wedge b\},\{a\cdot b\})$ fails.

We start with $R = \{(\{a \wedge b\}, \{a \cdot b\})\}$, and note that $\bar{\epsilon}(\{a \wedge b\}) = 0 = \bar{\epsilon}(\{a \cdot b\})$.
We now take a derivative w.r.t.\ $\alpha \in \at$ such that $\pi_\alpha\leqqba a \wedge b$.
To grow $R$ into a bisimulation up to congruence, all derivatives should be checked and possibly added to $R$; this specific choice, however, will lead to a counterexample.
We have that $\bar{\theta}(\{a \wedge b\}, \alpha)=\{ 1 \}$ and $\bar{\theta}(\{a \cdot b\}, \alpha) = \{1 \cdot b, 1\}$.
We add $(\{1\},\{1 \cdot b,1\})$ to $R$, noting that $\bar{\epsilon}(\{1\}) = 1 = \bar{\epsilon}(\{1 \cdot b, 1 \})$, and continue with the next derivative w.r.t.\ $\beta \in \at$ such that $\pi_\beta \leqqba b$.
We get $\bar{\theta}(\{1\}, \beta) = \emptyset$ and $\bar{\theta}(\{1 \cdot b, 1\}, \beta) = \{1\}$.
We now have $(\emptyset, \{1\}) \in R$, but $\bar{\epsilon}(\emptyset) \neq \bar{\epsilon}(\{1\})$.
Thus, we cannot construct a bisimulation up to congruence $R$ such that $(\{a \wedge b\},\{a \cdot b\}) \in R$.
\end{example}

\begin{remark}
For \KAO-expressions without observations, the derivatives with respect to an element of $\at$ result in the empty set.
Hence, for these \KAO-expressions, deciding equivalence comes down to standard derivative-based techniques for rational expressions~\cite{rutten-1998}.
\end{remark}

\section{Related work}%
\label{section:related-work}

This work fits in the larger tradition of Kleene algebra as a presentation of the ``laws of programming'', the latter having been studied by Hoare and collaborators~\cite{hoare-hayes-he-etal-1987,hoare-2014}.
More precisely, our efforts can be grouped with recent efforts to extend Kleene algebra with concurrency~\cite{hoare-moeller-struth-wehrman-2009,hoare-2014,laurence-struth-2017-arxiv,kappe-brunet-silva-zanasi-2018}, and thence with Boolean guards~\cite{jipsen-moshier-2016}.

We proved that $\equivkao$ is sound and complete w.r.t.\ $\semkao{-}$, based on the existing completeness proof for KA\@.
By re-using completeness results for a simpler algebra, the proof shows a clear separation of concerns between the ``old'' algebra being extended and the new layer of axioms placed on top.
This strategy pops up quite often in some form~\cite{kozen-smith-1996,anderson-foster-guha-etal-2014,laurence-struth-2017-arxiv,kappe-brunet-silva-zanasi-2018}.
We note that unlike~\cite{kozen-smith-1996,anderson-foster-guha-etal-2014}, our transformation does not proceed by induction on the term, but leverages the least fixpoints computable in every Kleene algebra.
Further, the combination of KA with additional hypotheses presented in~\cite{kozen-mamouras-2014} might yield another route to completeness.

The use of linear systems to study automata was pioneered by Conway~\cite{conway-1971} and Backhouse~\cite{backhouse-1975}.
Kozen's completeness proof expanded on this by generalising Kleene algebra to matrices~\cite{kozen-1994}.
The connection between linear systems, derivatives and completeness was studied by Kozen~\cite{kozen-2001} and Jacobs~\cite{jacobs-2006}.
To keep our presentation simple, we give our proof of completeness in elementary terms; a proof in terms of coalgebra~\cite{rutten-2000} may yet be possible.

Using bisimulation to decide language equivalence in a deterministic finite automaton originates from Hopcroft and Karp~\cite{hopcroft-karp-1971}.
This technique has many generalisations~\cite{rot-bonchi-bonsangue-etal-2017}.%

\section{Conclusions and further work}%
\label{section:further-work}

Kleene algebra with observations (KAO) is an algebraic framework that adds Boolean guards to Kleene algebra in such a way that a sensible extension with concurrency is still possible, in contrast with Kleene algebra with tests (KAT).
Indeed, the laws of KAO prevent the problem presented by \cref{antinomy:problem}.
The axiomatisation of KAO, as well as the decision procedure for equivalence, give an alternative foundation for combining KAO with concurrent Kleene algebra to arrive at a new equational calculus of concurrent programs.

The most obvious direction of further work is to define \emph{concurrent Kleene algebra with observations} (\emph{CKAO}) as the amalgamation of axioms of KAO and CKA, in pursuit of a characterisation of its equational theory and an analogous decidability result.
We conjecture that languages of \emph{sp-pomsets}~\cite{grabowski-1981,gischer-1988,lodaya-weil-2000} over actions and atoms, closed under some suitable relation, form the free model; a proof would likely build on~\cite{laurence-struth-2014} and re-use techniques from~\cite{kappe-brunet-silva-zanasi-2018}.
While the equational theory of CKAO might be decidable, we are less optimistic about a feasible algorithm; deciding the equational theory of CKA is already \textsc{expspace}-complete~\cite{brunet-pous-struth-2017}.

Another avenue of future research would be to create a programming language using KAO (or, possibly, CKAO) by instantiating actions and observations, and adding axioms that encode the intention of those primitives.
NetKAT~\cite{anderson-foster-guha-etal-2014,foster-kozen-milano-silva-thompson-2015,smolka-eliopoulos-foster-guha-2015}, a language for describing and reasoning about network behaviour is an excellent example of such an endeavour based on KAT\@.
Our long-term hope is that CKAO will function as a foundation for a ``concurrent'' version of NetKAT aimed at describing and reasoning about networks with concurrency.

Finally, we note that the decision procedure for KAO based on bisimulation up to congruence leaves room for optimisation.
Besides adapting the work on symbolic algorithms for bisimulation-based algorithms in KAT~\cite{pous-2014}, the transitivity property witnessed in \cref{lemma:double-derivatives} seems like it could sometimes allow a bisimulation-based algorithm to decide early.

\bibliographystyle{plainurl}
\bibliography{bibliography}

\ifarxiv%

\clearpage%
\appendix%

\section{Proofs about the language model}

\wkaosoundness*
\begin{proof}
It suffices to check the claim for the pairs generating $\equivwkao$.

For the pairs that come from $\equivba$, the claim then goes through by the fact that $\semwkao{-}$ coincides with $\semba{-}$ on $\termsprop$, and that $\equivba$ is sound w.r.t. $\semba{-}$ by \cref{theorem:ba-completeness}.

For the pairs that come from $\equivka$, the claim goes through by the observation that $2^{\Gamma^\star}$ forms a Kleene algebra when equipped with the operators used to define $\semwkao{-}$, and the fact that $\equivka$ encodes the axioms of a Kleene algebra.

Lastly, we can easily check that for any $p, q \in \termsprop$ we have that
\[
\semwkao{p \vee q}
    = \semba{p} \cup \semba{q}
    = \semwkao{p + q}
\quad\quad
\semwkao{\bot}
    = \semba{\bot}
    = \emptyset
    = \semwkao{0}
    \qedhere
    \ .
\]
\end{proof}

\closurealmostcommutes*
\begin{proof}
Let $\contr^*$ denote the reflexive, transitive closure of $\contr$.
\begin{enumerate}[(i)]
    \item
    First, suppose $w \in \semkao{e + f}$.
    Then $w \contr^* x$ for some $x \in \semwkao{e + f}$.
    If $x \in \semwkao{e}$, then $w \in \semkao{e}$; similarly, if $x \in \semwkao{f}$, then $w \in \semkao{f}$.

    For the other inclusion, suppose that $w \in \semkao{e} \cup \semkao{f}$.
    If $w \in \semkao{e}$, then $w \contr^* x$ for some $x \in \semwkao{e}$.
    In that case, $x \in \semwkao{e + f}$, and thus $w \in \semkao{e + f}$.
    Similarly, if $w \in \semkao{f}$, also $w \in \semkao{e + f}$.

    \item
    For the inclusion from left to right, we note that
    \[
        \semkao{e \cdot f}
            = \down{(\semwkao{e} \cdot \semwkao{f})}
            \subseteq \down{(\semkao{e} \cdot \semkao{f})}
            \ .
    \]

    For the other inclusion, suppose $w \in \down{(\semkao{e} \cdot \semkao{f})}$.
    Then there exist $x \in \semkao{e}$ and $y \in \semkao{f}$ such that $w \contr^* xy$.
    In that case, there also exist $u \in \semwkao{e}$ and $v \in \semwkao{f}$ such that $x \contr^* u$ and $y \contr^* v$.
    Hence, we find that $w \contr^* xy \contr^* xv \contr^* uv \in \semkao{e \cdot f}$, and thus $w \in \semkao{f}$.

    \item
    When $L$ is an observation language, we define $L^n$ by setting $L^0 = \{ \epsilon \}$ and $L^{n+1} = L \cdot L^n$.
    It is then straightforward to see that $L^\star = \bigcup_{n \in \naturals} L^n$.
    Similarly, we define $e^n$ by $e^0 = 1$ and $e^{n+1} = e \cdot e^n$.
    By~(i) and~(ii), we find:
    \begin{equation*}
        \semkao{e^\star}
            = \down{\semwkao{e}^\star}
            = \bigcup_{n \in \naturals} \down{\semwkao{e}^n}
            = \bigcup_{n \in \naturals} \semkao{e^n}
            = \bigcup_{n \in \naturals} \down{\semkao{e}^n}
            = \down{\semkao{e}^\star}
            \ .
        \qedhere
    \end{equation*}
\end{enumerate}
\end{proof}

\soundness*
\begin{proof}
We proceed by induction on the construction of $\equivkao$.
In the base, we check the pairs that generate $\equivkao$.
For the pairs that come from $\equivwkao$, we know by the above that $\semwkao{e} = \semwkao{f}$, and hence $\semkao{e} = \down{\semwkao{e}} = \down{\semwkao{f}} = \semkao{f}$.

For the contraction law, let $p, q \in \termsprop$; we should show $\semkao{p \wedge q} \subseteq \semkao{p \cdot q}$.
To see this, let $w \in \semkao{p \wedge q}$.
Then there exists an $\alpha \in \semwkao{p \wedge q} = \semba{p} \cap \semba{q}$ with $w \contr \alpha$.
Since $\alpha\alpha \in \semwkao{p \cdot q}$ and $\alpha \contr \alpha\alpha$, it follows that $\alpha \in \semkao{p \cdot q}$ and hence $w \in \semkao{p \cdot q}$.

For the inductive step, we should verify that the claim is preserved by the closure rules for congruence.
This gives us three cases to consider.
\begin{itemize}
    \item
    If $e = e_0 + e_1$ and $f = f_0 + f_1$ with $e_0 \equivkao f_0$ and $e_1 \equivkao f_1$, then by induction we know that $\semkao{e_0} = \semkao{f_0}$ and $\semkao{e_1} = \semkao{f_1}$.
    By \cref{lemma:closure-almost-commutes}(i), we can then derive that $\semkao{e} = \semkao{e_0} \cup \semkao{e_1} = \semkao{f_0} \cup \semkao{f_1} = \semkao{f}$. 

    \item
    If $e = e_0 \cdot e_1$ and $f_0 \cdot f_1$ with $e_0 \equivkao f_0$ and $e_1 \equivkao f_1$, then by induction we know that $\semkao{e_0} = \semkao{f_0}$ and $\semkao{e_1} = \semkao{f_1}$.
    By \cref{lemma:closure-almost-commutes}(ii), we can then derive that $\semkao{e} = \down{(\semkao{e_0} \cdot \semkao{e_1})} = \down{(\semkao{f_0} \cdot \semkao{f_1})} = \semkao{f}$ 

    \item
    If $e = e_0^\star$ and $f = f_0^\star$ and $e_0 \equivkao f_0$, then by induction we know that $\semkao{e_0} = \semkao{f_0}$.
    By \cref{lemma:closure-almost-commutes}(iii), we can then derive that $\semkao{e} = \down{\semkao{e_0}^\star} = \down{\semkao{f_1}^\star} = \semkao{f}$. 
    \qedhere
\end{itemize}
\end{proof}

\section{Proofs about partial derivatives}

\accvsunit*
\begin{proof}
First, assume that $\epsilon \in \semkao{e}$.
We will now show that $\epsilon(e) = 1$ by induction on $e$.
In the base, it suffices to consider the case where $e = 1$; here, we find $\epsilon(e) = 1$ by definition of $\epsilon$.
For the inductive step, there are three cases to consider.
\begin{itemize}
    \item
    If $e = e_0 + e_1$, then $\epsilon \in \semkao{e_0}$ or $\epsilon \in \semkao{e_1}$ by \cref{lemma:closure-almost-commutes}(i). 
    By induction, we then find that $\epsilon(e_0) = 1$ or $\epsilon(e_1) = 1$, and hence $\epsilon(e) = 1$.

    \item
    If $e = e_0 \cdot e_1$, then $\epsilon \contr^* wx$ with $w \in \semkao{e_0}$ and $x \in \semkao{e_1}$ by \cref{lemma:closure-almost-commutes}(ii). 
    By definition of $\contr$, we find that $w = x = \epsilon$.
    By induction, we then have that $\epsilon(e_0) = \epsilon(e_1) = 1$, and thus $\epsilon(e) = 1$.

    \item
    If $e = e_0^\star$, then $\epsilon(e) = 1$ by definition.
\end{itemize}

Next, note that to prove that $\epsilon(e) = 1$ implies $\epsilon \in \semkao{e}$, it suffices to prove that $\epsilon(e) \leqqkao e$, by \cref{lemma:soundness}.
This we do by induction on $e$.
In the base, the claim is true trivially.
For the inductive step, there are three cases to consider.
\begin{itemize}
    \item
    If $e = e_0 + e_1$, then we can derive that by induction that
    \[
        \epsilon(e) \equivkao \epsilon(e_0) + \epsilon(e_1) \leqqkao e_0 + e_1 = e \ .
    \]

    \item
    If $e = e_0 \cdot e_1$, then we can derive by induction that
    \[
        \epsilon(e) = \epsilon(e_0) \cdot \epsilon(e_1) \leqqkao e_0 \cdot e_1 = e \ .
    \]

    \item
    If $e = e_0^\star$, then $\epsilon(e) = 1 \leqqkao 1 + e_0 \cdot e_0^\star \equivkao e_0^\star = e$.
    \qedhere
\end{itemize}
\end{proof}

\doublederivatives*
\begin{proof}
We proceed by induction on $e$.
In the base, where $e \in \{ 0, 1 \} \cup \Sigma \cup \termsprop$, the claim holds vacuously.
For the inductive step, there are three cases to consider.
\begin{itemize}
    \item
    If $e = e_0 + e_1$, then $e' \in \zeta(e_0, \alpha)$ or $e' \in \zeta(e_1, \alpha)$; we assume the former without loss of generality.
    By induction, we then know that $\zeta(e', \alpha) \subseteq \zeta(e_0, \alpha)$; since $\zeta(e_0, \alpha) \subseteq \zeta(e, \alpha)$, the claim then follows.

    \item
    If $e = e_0 \cdot e_1$, we have three subcases to consider.
    \begin{itemize}
        \item
        If $e' = e_0' \cdot e_1$ with $e_0' \in \zeta(e_0, \alpha)$, then by induction we know that $\zeta(e_0', \alpha) \subseteq \zeta(e_0, \alpha)$.
        If $e'' \in \zeta(e_0' \cdot e_1, \alpha)$, then we have two more subcases to consider.
        \begin{itemize}
            \item
            If $e'' = e_0'' \cdot e_1$ for some $e_0'' \in \zeta(e_0', \alpha)$, then by induction we know that $e_0'' \in \zeta(e_0, \alpha)$, and therefore $e'' \in \zeta(e_0 \cdot e_1, \alpha)$.

            \item
            If $e'' \in Z(e_0', e_1, \alpha) = \zeta(e_1, \alpha)$, then one of two cases applies.
            First, if $\epsilon(e_0') = 1$, we find that $e'' \in \zeta(e_1, \alpha) = Z(e_0, e_1, \alpha)$.
            Second, if there exists an $e_0'' \in \zeta(e_0', \alpha)$ such that $\epsilon(e_0'') = 1$, then we find that $e_0'' \in \zeta(e_0, \alpha)$ by induction; hence, $e'' \in \zeta(e_1, \alpha) = Z(e_0, e_1, \alpha) \subseteq \zeta(e, \alpha)$.
        \end{itemize}

        \item
        If $e' \in Z(e_0, e_1, \alpha) = \zeta(e_1, \alpha)$ then by induction we find $e'' \in \zeta(e_1, \alpha) = Z(e_0, e_1, \alpha) \subseteq \delta(e, \alpha)$.
    \end{itemize}

    \item
    If $e = e_0^\star$, then $e' = e_0' \cdot e_0^\star$ for some $e_0' \in \zeta(e_0, \alpha)$.
    If $e'' \in \zeta(e', \alpha)$, then we have two subcases to consider.
    \begin{itemize}
        \item
        If $e'' = e_0'' \cdot e_0^\star$ for some $e_0'' \in \zeta(e_0', \alpha)$, then by induction we know that $e_0'' \in \zeta(e_0, \alpha)$.
        It therefore follows that $e'' \in \zeta(e, \alpha)$.

        \item
        If $e'' \in Z(e_0', e_0^\star, \alpha) = \zeta(e_0^\star, \alpha)$, then $e'' \in \zeta(e, \alpha)$ immediately.
        \qedhere
    \end{itemize}
\end{itemize}
\end{proof}

\halffundamentaltheorem*
\begin{proof}
We prove the second claim by induction on $e$; the first claim can be shown analogously.
For the base, the claim holds vacuously if $e \in \{ 0, 1 \} \cup \Sigma$.
When $e \in \termsprop$, we have $e' = 1$ and $\pi_\alpha \leqqba e$.
Hence, we conclude that $\pi_\alpha \cdot e' \equivkao \pi_\alpha \leqqkao e$.

For the inductive step, there are three cases to consider.
\begin{itemize}
    \item
    If $e = e_0 + e_1$, then either $e' \in \zeta(e_0, \alpha)$ or $e' \in \zeta(e_1, \alpha)$; w.l.o.g., we assume the former.
    By induction, we then find that $\pi_\alpha \cdot e' \leqqkao e_0 \leqqkao e$.

    \item
    If $e = e_0 \cdot e_1$, then there are two cases to consider.
    \begin{itemize}
        \item
        If $e' = e_0' \cdot e_1$ for some $e_0' \in \zeta(e_0, \alpha)$, then by induction we know that $\pi_\alpha \cdot e_0' \leqqkao e_0$.
        By induction, we derive $\pi_\alpha \cdot e' \equivkao \pi_\alpha \cdot e_0' \cdot e_1 \leqqkao e_0 \cdot e_1 = e$.

        \item
        If $e' \in Z(e_0, e_1, \alpha) = \zeta(e_1, \alpha)$, then $\pi_\alpha \cdot e' \leqqkao e_1$, and one of two cases applies.
        First, if $\epsilon(e_0) = 1$, then $\pi_\alpha \cdot e' \leqqkao e_1 \equivkao 1 \cdot e_1 \leqqkao e_0 \cdot e_1 = e$, by \cref{lemma:acc-vs-unit}.
        Second, if there exists some $e_0' \in \zeta(e_0, \alpha)$ with $\epsilon(e_0') = 1$, then by induction $\pi_\alpha \cdot e_0' \leqqkao e_0$.
        Hence, we can derive that
        \begin{align*}
        \pi_\alpha \cdot e'
            &\equivkao (\pi_\alpha \wedge \pi_\alpha) \cdot e' \\
            &\leqqkao (\pi_\alpha \cdot \pi_\alpha) \cdot e' \\
            &\equivkao (\pi_\alpha \cdot 1) \cdot (\pi_\alpha \cdot e') \\
            &\leqqkao (\pi_\alpha \cdot e_0') \cdot (\pi_\alpha \cdot e') \\
            &\leqqkao e_0 \cdot e_1 \\
            &\equivkao e
             \ .
        \end{align*}
    \end{itemize}

    \item
    If $e = e_0^\star$, then $e' = e_0' \cdot e_0^\star$ for some $e_0' \in \zeta(e_0, \alpha)$.
    By induction, we know that $\pi_\alpha \cdot e_0' \leqqkao e_0$.
    We can then derive that
    \[
        \pi_\alpha \cdot e'
            \equivkao \pi_\alpha \cdot e_0' \cdot e_0^\star
            \leqqkao e_0 \cdot e_0^\star
            \leqqkao 1 + e_0 \cdot e_0^\star
            \equivkao e_0^\star
            \ .
            \qedhere
    \]
\end{itemize}
\end{proof}

\reachcontainsderivatives*
\begin{proof}
We prove the second claim by induction on $e$; the first claim can be proved analogously.
In the base, there are two cases to consider.
If $e \in \{0, 1\} \cup \Sigma$, then the claim holds vacuously.
Otherwise, if $e \in \termsprop$, then $\zeta(e, \alpha) \subseteq \{ 1 \} \subseteq \rho(e)$.

For the inductive step, there are three cases to consider.
\begin{itemize}
    \item
    If $e = e_0 + e_1$, then by induction we have $\zeta(e_0, \alpha) \subseteq \rho(e_0)$ and $\zeta(e_1, \alpha) \subseteq \rho(e_1)$.
    Hence, we find that $\zeta(e, \alpha) = \zeta(e_0, \alpha) \cup \zeta(e_1, \alpha) \subseteq \rho(e_0) \cup \rho(e_1) = \rho(e)$.

    \item
    If $e = e_0 \cdot e_1$, then by induction we have $\zeta(e_0, \alpha) \subseteq \rho(e_0)$ and $\zeta(e_1, \alpha) \subseteq \rho(e_1)$.
    Hence, we can calculate that
    \[
    \zeta(e, \alpha)
        = \{ e_0' \cdot e_1 : e_0' \in \zeta(e_0, \alpha) \} \cup Z(e_0, e_1, \alpha)
        \subseteq \{ e_0' \cdot e_1 : e_0' \in \rho(e_0) \} \cup \rho(e_1)
         = \rho(e)
         \ .
    \]

    \item
    If $e = e_0^\star$, then by induction we have $\zeta(e_0, \alpha) \subseteq \rho(e_0)$.
    Hence, we find that
    \[
        \zeta(e, \alpha) = \{ e_0' \cdot e_0^\star : e_0' \in \zeta(e_0, \alpha) \} \subseteq \{ e_0' \cdot e_0^\star : e_0' \in \rho(e_0) \} \subseteq \rho(e) \ .
    \]
\end{itemize}

For the second part, we prove that if $e' \in \rho(e)$, then $\rho(e') \subseteq \rho(e)$; this, in conjunction with the first part, validates the claim.
We proceed by induction on $e$.
In the base, there are two cases to consider.
First, if $e \in \{ 0, 1 \} \cup \Sigma$, then the claim holds vacuously.
Otherwise, if $e \in \termsprop$, then $e' = 1$, and therefore $\rho(e') = \emptyset \subseteq \rho(e)$.

For the inductive step, there are three cases to consider.
\begin{itemize}
    \item
    If $e = e_0 + e_1$, assume w.l.o.g.\ $e' \in \rho(e_0)$.
    By induction, $\rho(e') \subseteq \rho(e_0) \subseteq \rho(e)$.

    \item
    If $e = e_0 \cdot e_1$ then there are two cases to consider.
    \begin{itemize}
        \item
        If $e' = e_0' \cdot e_1$ where $e_0' \in \rho(e_0)$, then we calculate
        \[
            \rho(e')
                = \{ e_0'' \cdot e_1 : e_0'' \in \rho(e_0') \} \cup \rho(e_1)
                \subseteq \{ e_0'' \cdot e_1 : e_0'' \in \rho(e_0) \} \cup \rho(e_1)
                = \rho(e)
                \ .
        \]
        \item
        If $e' \in \rho(e_1)$, then by induction we have $\rho(e') \subseteq \rho(e_1) \subseteq \rho(e)$.
    \end{itemize}

    \item
    If $e = e_0^\star$, then either $e' = 1$ or $e' = e_0' \cdot e_0^\star$ for some $e_0' \in \rho(e_0)$.
    In the former case, $\rho(e') = \emptyset$.
    In the latter case, we find by induction that
    \[
        \rho(e')
            = \{ e_0'' \cdot e_0^\star : e_0'' \in \rho(e_0') \} \cup \rho(e_0^\star)
            \subseteq \{ e_0'' \cdot e_0^\star : e_0'' \in \rho(e_0) \} \cup \rho(e_0^\star) = \rho(e)
            \ .
            \qedhere
    \]
\end{itemize}
\end{proof}

\reconstructfrominitial*
\begin{proof}
The proof proceeds by induction on $e$.
In the base, where $e \in \{ 0,  1 \} \cup \Sigma \cup \termsprop$, the claim holds trivially.
For the inductive step, there are three cases to consider.
\begin{itemize}
    \item
    If $e = e_0 + e_1$, then we calculate by induction that
    \[
        e
            = e_0 + e_1
            \equivkao \sum_{e_0' \in \iota(e_0)} e_0' + \sum_{e_1' \in \iota(e_1)} e_1
            \equivkao \sum_{e' \in \iota(e)} e'
            \ .
    \]

    \item
    If $e = e_0 \cdot e_1$, then we calculate by induction that
    \[
        e
            = e_0 \cdot e_1
            \equivkao \sum_{e_0' \in \iota(e_0)} e_0' \cdot e_1
            \equivkao \sum_{e' \in \iota(e)} e'
            \ .
    \]

    \item
    If $e = e_0^\star$, then we calculate by induction that
    \[
        e
            = e_0^\star
            \equivkao 1 + e_0 \cdot e_0^\star
            \equivkao 1 + \sum_{e_0' \in \iota(e_0')} e_0' \cdot e_0^\star
            \equivkao \sum_{e' \in \iota(e)} e'
            \ .
            \qedhere
    \]
\end{itemize}
\end{proof}

\section{Proofs towards completeness}

\kaocompletenesspartial*
\begin{proof}
Note that we can regard atomic guarded rational terms as rational expressions over $\Delta = \Sigma \cup \Pi$; let this interpretation be given by $\semka{-}: \termsagrat \to 2^{\Delta^\star}$.
It is now easy to show that for $g \in \termsagrat$ closed, we have $\semkao{g} = \semka{g}$.
By the premise that $e$ and $f$ are closed, and that $\semkao{e} = \semkao{f}$, we thus find that $\semka{e} = \semka{f}$.
By \cref{theorem:ka-completeness}, it then follows that $e \equivka f$; hence, $e \equivkao f$, since the axioms generating $\equivka$ are also valid for $\equivkao$.
\end{proof}

\leastsolutioncharacterisation*
\begin{proof}
Suppose that $s$ is a $\rho(e)$-vector that satisfies the system of equations above.
Using \cref{lemma:reach-contains-derivatives}, we can then calculate for any $e' \in \rho(e)$ as follows:
\begin{align*}
(M_e \cdot s + x_e \fatsemi f)(e')
    &\equivaka \epsilon(e') \cdot f + \sum_{e'' \in \rho(e)} M_e(e', e'') \cdot s(e'') \\
    &\equivaka \epsilon(e') \cdot f + \sum_{e'' \in \rho(e)} \Bigl( \sum_{e'' \in \delta(e', a)} a + \sum_{e'' \in \zeta(e', \alpha)} \pi_\alpha \Bigr) \cdot s(e'') \\
    &\equivaka \epsilon(e') \cdot f + \sum_{e'' \in \delta(e', a)} a \cdot s(e'') + \sum_{e' \in \zeta(e'', \alpha)} \pi_\alpha \cdot s(e'')
     \leqqaka s(e')
     \ .
\end{align*}
Hence, we find that $M_e \cdot s + x_e \fatsemi f \leqqaka s$.
Since $M_e^\star \cdot (x_e \fatsemi f) = (M_e^\star \cdot x_e) \fatsemi f = s_e \fatsemi f$ is the least $\rho(e)$-vector for which this holds, we can conclude that $s_e \fatsemi f \leqqaka s$.

Conversely, a similar derivation shows that $s_e \fatsemi f$ is can take the place of $s$ in the system above, and hence the least such $s$ is a lower bound for $s_e \fatsemi f$.

As for the other claim, let $s$ be the $\rho(e)$-vector defined as follows:
\[
    s(e') = \epsilon(e') \cdot f
        + \sum_{e'' \in \delta(e', a)} a \cdot (s_e \fatsemi f)(e'')
        + \sum_{e'' \in \zeta(e', \alpha)} \pi_\alpha \cdot (s_e \fatsemi f)(e'')
        \ .
\]
Now $s \leqqaka s_e \fatsemi f$ by the first part of this proposition.
By monotonicity, we also derive the following:
\begin{align*}
&\epsilon(e') \cdot f
    + \sum_{e'' \in \delta(e', a)} a \cdot s(e'')
    + \sum_{e'' \in \zeta(e', \alpha)} \pi_\alpha \cdot s(e'') \\
&\quad\leqqaka \epsilon(e') \cdot f
    + \sum_{e'' \in \delta(e', a)} a \cdot (s_e \fatsemi f)(e'')
    + \sum_{e'' \in \zeta(e', \alpha)} \pi_\alpha \cdot (s_e \fatsemi f)(e'')
 \equivaka s(e')
 \ ,
\end{align*}
meaning $s$ satisfies the system obtained from $e$ and $f$; hence, $s_e \fatsemi f \leqqaka s$ by the first part of this proposition again.
Consequently, $s \equivkao s_e \fatsemi f$, validating the claim.
\end{proof}

To prove \cref{proposition:least-solution-identity}, we need two technical lemmas.

\begin{restatable}{lem}{fixpointpropertiessubsystem}%
\label{lemma:fixpoint-properties-subsystem}
Let $e, f \in \termsgrat$.
If $f' \in \rho(f)$ and $\rho(f) \subseteq \rho(e)$, then $s_f(f') \leqqaka s_e(f')$.
\end{restatable}
\begin{proof}
We fix a $\rho(f)$-vector $s$ by choosing for $f' \in \rho(f)$ that $s(f') = s_e(f')$.
By \cref{proposition:least-solution-characterisation}, we have for any $f' \in \rho(f)$ that
\[
    \epsilon(f') + \sum_{f'' \in \delta(f', a)} a \cdot s(f'') + \sum_{f'' \in \zeta(f', \alpha)} \pi_\alpha \cdot s(f'') \leqqaka s(f') \ .
\]
Hence, $s$ satisfies the system of linear equations obtained from $f$, and therefore $s_f$ is a lower bound of $s$; the claim then follows.
\end{proof}

\begin{restatable}{lem}{fixpointpropertiescontinuation}%
\label{lemma:fixpoint-properties-continuation}
Let $e, f \in \termsgrat$.
If $e' \in \rho(e)$ and $f' \in \iota(f)$, then $s_e(e') \cdot s_f(f') \leqqaka s_{e \cdot f}(e' \cdot f)$.
\end{restatable}
\begin{proof}
Using \cref{proposition:least-solution-characterisation,lemma:reconstruct-from-initial,lemma:acc-vs-unit,lemma:fixpoint-properties-subsystem}, we calculate for all $e' \in \rho(e)$:
\begin{align*}
&\epsilon(e') \cdot s_f(f') + \sum_{e'' \in \delta(e', a)} a \cdot s_{e \cdot f}(e'' \cdot f) + \sum_{e'' \in \zeta(e', \alpha)} \pi_\alpha \cdot s_{e \cdot f}(e'' \cdot f) \\
&\quad\equivaka \epsilon(e') \cdot \Bigl( \epsilon(f') + \sum_{f'' \in \delta(f', a)} a \cdot s_f(f') + \sum_{f'' \in \zeta(f', \alpha)} \pi_\alpha \cdot s_f(f') \Bigr) \\
&\quad\phantom{\equivaka} + \sum_{e'' \in \delta(e', a)} a \cdot s_{e \cdot f}(e'' \cdot f) + \sum_{e'' \in \zeta(e', \alpha)} \pi_\alpha \cdot s_{e \cdot f}(e'' \cdot f) \\
&\quad\leqqaka \epsilon(e' \cdot f) + \sum_{g \in \delta(e' \cdot f, a)} a \cdot s_{e \cdot f}(g) + \sum_{g \in \zeta(e' \cdot f, \alpha)} \pi_\alpha \cdot s_{e \cdot f}(g)
 \equivaka s_{e \cdot f}(e' \cdot f)
 \ .
\end{align*}
Hence, the $\rho(e)$-vector $s$ where $s(e') = s_{e \cdot f}(e' \cdot f)$ satisfies the system obtained from $e$ and $s_f(f')$.
By \cref{proposition:least-solution-characterisation}, we find for all $e' \in \rho(e)$ that
\[
    s_e(e') \cdot s_f(f') = (s_e \fatsemi s_f(f'))(e') \leqqaka s(e') = s_{e \cdot f}(e' \cdot f)
    \ .
    \qedhere
\]
\end{proof}

\leastsolutionidentity*
\begin{proof}
Let $s$ be the identity on $\rho(e)$.
By \cref{lemma:half-fundamental-theorem} and \cref{lemma:acc-vs-unit}, we find that $s$ satisfies the system of linear equations obtained from $e$ as in \cref{proposition:least-solution-characterisation}.
Hence, by that lemma, for all $e' \in \rho(e)$ we have $s_e(e') \leqqkao s(e') = e'$.

The other direction goes by induction on $e$.
In the base, there are four cases.
\begin{itemize}
    \item
    If $e = 0$, then $\rho(e) = \emptyset$, and so the claim holds trivially.

    \item
    If $e = 1$, then $e' = 1$, and $e' = 1 = \epsilon(1) \leqqkao s_e(1) = s_e(e')$.

    \item
    If $e = a$ for some $a \in \Sigma$, then either $e' = a$, or $e' = 1$.
    The case where $e' = 1$ can be argued as before.
    On the other hand, if $e' = a$, then we find that
    \[
        e' = a \equivkao a \cdot 1 \leqqkao a \cdot s_e(1)
           \leqqkao s_e(a)
           = s_e(e')
           \ .
    \]

    \item
    If $e = p$ for some $p \in \termsprop$, then either $e' = p$ or $e' = 1$.
    As before, we have to argue the case where $e' = p$ only.
    By \cref{theorem:ba-completeness} and \cref{lemma:atom-semantics}, we find
    \[
        e' = p
            \equivkao \sum_{\pi_\alpha \leqqba p} \pi_\alpha
            \leqqkao \sum_{\pi_\alpha \leqqba p} \pi_\alpha \cdot s_e(1)
            \leqqkao s_e(p)
            = s_e(e')
            \ .
    \]
\end{itemize}
For the inductive step, there are three cases
\begin{itemize}
    \item
    If $e = e_0 + e_1$, then the claim follows by \cref{lemma:fixpoint-properties-subsystem}

    \item
    If $e = e_0 \cdot e_1$, then either $e' = e_0' \cdot f$ for some $e' \in \rho(e_0)$, or $e' \in \rho(e_1)$.
    In the former case, we find by induction, \cref{lemma:fixpoint-properties-continuation,lemma:reconstruct-from-initial} that
    \[
        e'
            = e_0' \cdot e_1
            \equivkao \sum_{e_1' \in \iota(e_1)} e_0' \cdot e_1'
            \leqqkao \sum_{e_1' \in \iota(e_1)} s_{e_0}(e_0') \cdot s_{e_1}(e_1')
            \leqqkao s_e(e')
            \ .
    \]
    In the latter case, we find $e' \leqqkao s_e(e')$ as before.

    \item
    If $e = e_0^\star$, then it suffices to treat the case where $e' = e_0' \cdot e$ with $e_0' \in \rho(e_0)$ only.
    By \cref{lemma:fixpoint-properties-continuation}, we derive for any $e'' \in \iota(e)$ that $s_{e_0}(e_0') \cdot s_e(e') \leqqkao s_{e_0 \cdot e}(e_0' \cdot e)$.
    By \cref{lemma:fixpoint-properties-subsystem} and the observation that $\rho(e_0 \cdot e) \subseteq \rho(e)$, we furthermore know that $s_{e_0 \cdot e}(e_0' \cdot e) \leqqkao s_e(e_0' \cdot e)$.
    We then conclude by induction, \cref{lemma:reconstruct-from-initial} and the above that
    \begin{equation*}
    e'
        = e_0' \cdot e
        \equivkao \sum_{e'' \in \iota(e)} e_0' \cdot e''
        \leqqkao \sum_{e' \in \iota(e)} s_{e_0}(e_0') \cdot s_e(e')
        \leqqkao s_e(e_0' \cdot e)
        = s_e(e')
        \ .
        \qedhere
    \end{equation*}
\end{itemize}
\end{proof}

\leastsolutionidentitytotal*
\begin{proof}
Using \cref{proposition:least-solution-identity,lemma:reconstruct-from-initial}, we calculate:
\[
    e
        \equivkao \sum_{e' \in \iota(e)} e'
        \equivkao \sum_{e' \in \iota(e)} s_e(e')
        \equivkao \hat{e}
        \ .
        \qedhere
\]
\end{proof}

\closureequivalence*
\begin{proof}
It suffices to prove that, whenever $w \contr x$, also $w \altcontr x$, and when $w \altcontr x$ also $w \contr^* x$, where $\contr^*$ denotes the reflexive and transitive closure of $\contr$.
For the first part, suppose that $w \contr x$; in that case, there exist $u, v \in \Gamma^\star$ such that $\alpha \in \at$, with $w = u\alpha{}v$ and $x = u\alpha\alpha{}v$.
A straightforward inductive argument on the length of $v$ then shows that $v \altcontr v$, and hence $\alpha{}v \altcontr \alpha\alpha{}v$.
By induction on the length of $w$, we conclude $w = u\alpha{}v \altcontr u\alpha\alpha{}v = x$.

For the other direction, suppose that $w \altcontr x$; we proceed by induction on the construction of $\altcontr$.
In the base, we have that $w = \epsilon = x$, in which case $w \contr^* x$ immediately.
For the inductive step, there are two cases to consider.
\begin{itemize}
    \item
    If there exists an $\mathfrak{a} \in \Gamma$ such that $w = \mathfrak{a}w'$ and $x = \mathfrak{a}x'$ with $w' \altcontr x'$, then by induction $w' \contr^* x'$.
    There must then exist $y_0, \dots, y_{n-1} \in \Gamma^\star$ such that $w' = y_0 \contr y_1 \contr \dots \contr y_{n-1} = x'$.
    By induction on $n$, we can then show that $w = \mathfrak{a}w' = \mathfrak{a}y_0 \contr \mathfrak{a}y_1 \contr \dots \contr \mathfrak{a}y_{n-1} = \mathfrak{a}x' = x$.

    \item
    If there exists an $\alpha \in \at$ such that $w = \alpha{}w'$ and $x = \alpha\alpha{}x'$ with $w' \altcontr x'$, then by induction $w' \contr^* x'$.
    By an argument similar to the previous case, we can show that $w = \alpha{}w' \contr^* \alpha{}x'$.
    Since $\alpha{}x' \contr \alpha\alpha{}x' = x$, it follows that $w \contr^* x$.
    \qedhere
\end{itemize}
\end{proof}

\componentclosure*
\begin{proof}
By definition, we have that $\semwkao{s_e(e')} \subseteq \semkao{s_e(e')}$.
For the reverse inclusion, it suffices to prove that $\semwkao{s_e(e')}$ is closed under~$\altcontr$, by \cref{lemma:closure-equivalence}.
More precisely, if $x \in \semwkao{s_e(e')}$ and $w \altcontr x$, we should show that $w \in \semwkao{s_e(e')}$.
We do this in tandem for all $e' \in \rho(e)$, and proceed by induction on $\altcontr$.
In the base, we have that $w = \epsilon = x$, and so the claim follows.

For the inductive step, there are two cases to consider.
\begin{itemize}
    \item
    If there exists an $\mathfrak{a} \in \Gamma$ such that $w = \mathfrak{a}w'$ and $x = \mathfrak{a}x'$ with $w' \altcontr x'$, then either $\mathfrak{a} \in \Sigma$ or $\mathfrak{a} \in \at$.
    In the latter case, we find by \cref{proposition:least-solution-characterisation} (for $\AKA = \WKAO$) and \cref{lemma:wkao-soundness} that there exist $\alpha \in \at$ and $e'' \in \zeta(e', \alpha)$ such that $\mathfrak{a} \in \semba{\pi_\alpha}$ and $x' \in \semwkao{s_e(e'')}$.
    By induction, $w' \in \semwkao{s_e(e'')}$; but then, again by \cref{proposition:least-solution-characterisation} and \cref{lemma:wkao-soundness} it follows that $w = \mathfrak{a}w' \in \semwkao{s_e(e')}$.
    The case where $\mathfrak{a} \in \Sigma$ can be argued analogously.

    \item
    Suppose there exists an $\alpha \in \at$ such that $w = \alpha{}w'$ and $x = \alpha\alpha{}x'$ with $w' \altcontr x'$.
    As in the previous case, we find $\beta, \gamma \in \at$ and $e'' \in \zeta(e', \beta)$ and $e''' \in \zeta(e'', \gamma)$ such that $\alpha \in \semba{\pi_\beta}$ and $\alpha \in \semba{\pi_\gamma}$, as well as $x' \in \semwkao{s_e(e''')}$. 
    By induction, we find $w' \in \semwkao{s_e(e''')}$. 
    By \cref{lemma:atom-semantics}, it also follows that $\beta = \alpha = \gamma$; hence, by \cref{lemma:double-derivatives}, we find $e''' \in \zeta(e', \alpha)$. 
    Again applying \cref{proposition:least-solution-characterisation} and \cref{lemma:wkao-soundness}, we can then conclude that $w = \alpha{}w' \in \semwkao{s_e(e')}$.
    \qedhere
\end{itemize}
\end{proof}

\closure*
\begin{proof}
Using \cref{lemma:soundness,lemma:closure-almost-commutes,proposition:component-closure,lemma:wkao-soundness}, we find
\[
    \semkao{\hat{e}}
        = \bigcup_{e' \in \iota(e)} \semkao{s_e(e')}
        = \bigcup_{e' \in \iota(e)} \semwkao{s_e(e')}
        = \semwkao{\hat{e}}
        \ .
        \qedhere
\]
\end{proof}

\section{Proofs about the decision procedure}

For the proofs that follow, we shall need the following technical lemma.

\begin{restatable}{lem}{iotaderivativescontainment}%
\label{lemma:iota-derivatives-containment}
Let $e \in \termsgrat$.
If $a \in \Sigma$ and $e' \in \iota(e)$, then $\delta(e', a) \subseteq \delta(e, a)$.
Furthermore, if $\alpha \in \at$ and $e' \in \iota(e)$, then $\zeta(e', \alpha) \subseteq \zeta(e, \alpha)$.
\end{restatable}
\begin{proof}
We prove the second claim by induction on $e$; the first claim can be shown similarly.
In the base, where $e \in \{ 0, 1 \} \cup \Sigma \cup \termsprop$, the claim holds immediately.
For the inductive step, there are three cases to consider.
\begin{itemize}
    \item
    If $e = e_0 + e_1$, assume w.l.o.g.\ that $e' \in \iota(e_0)$.
    By induction, we then derive
    \[
        \zeta(e', \alpha)
            \subseteq \zeta(e_0, \alpha)
            \subseteq \zeta(e, \alpha)
            \ .
    \]

    \item
    If $e = e_0 \cdot e_1$, then $e' = e_0' \cdot e_1$ for some $e_0' \in \iota(e_0)$.
    We now claim that $Z(e_0', e_1, \alpha) \subseteq Z(e_0, e_1, \alpha)$.
    To see this, it suffices to consider the case where $Z(e_0', e_1, \alpha) = \zeta(e_1, \alpha)$, i.e., where $\epsilon(e_0') = 1$ or there exists an $e_0'' \in \zeta(e_0', \alpha)$ with $\epsilon(e_0'') = 1$.
    In the former case, $\epsilon(e_0) = 1$ by \cref{lemma:acc-vs-unit} and \cref{lemma:reconstruct-from-initial}.
    In the latter case, we find that $e_0'' \in \zeta(e_0, \alpha)$ by induction.
    Thus, we have $Z(e_0', e_1, \alpha) = \zeta(e_1, \alpha) = Z(e_0, e_1, \alpha)$.
    Hence, we derive by induction that
    \begin{align*}
    \zeta(e', \alpha)
        &= \{ e_0'' \cdot e_1 : e_0'' \in \zeta(e_0', \alpha) \} \cup Z(e_0', e_1, \alpha) \\
        &\subseteq \{ e_0'' \cdot e_1 : e_0'' \in \zeta(e_0, \alpha) \} \cup Z(e_0, e_1, \alpha)
         = \zeta(e, \alpha)
         \ .
    \end{align*}

    \item
    If $e = e_0^\star$, then $e' = e_0' \cdot e_0^\star$ for some $e_0' \in \iota(e_0)$.
    By induction, we then derive
    \begin{align*}
    \zeta(e', \alpha)
        &= \{ e_0'' \cdot e_0^\star : e_0'' \in \zeta(e_0', \alpha) \} \cup \zeta(e_0^\star, \alpha) \\
        &\subseteq \{ e_0'' \cdot e_0^\star : e_0'' \in \zeta(e_0, \alpha) \} \cup \zeta(e_0^\star, \alpha)
         = \zeta(e, \alpha)
         \ .
         \qedhere
    \end{align*}
\end{itemize}
\end{proof}

\fundamental*
\begin{proof}
Let us write $\tilde{e}$ for the right-hand side of the claimed equivalence.
By \cref{lemma:acc-vs-unit} and \cref{lemma:half-fundamental-theorem}, we already know that $\tilde{e} \leqqkao e$.

For the other direction, we derive as follows
\begin{align*}
e
    &\equivkao \sum_{e' \in \iota(e)} s_e(e')
        \tag{\cref{corollary:least-solution-identity-total}} \\
    &\equivkao \sum_{e' \in \iota(e)} \Bigl( \epsilon(e') + \sum_{e'' \in \delta(e', a)} a \cdot s_e(e'') + \sum_{e'' \in \zeta(e', \alpha)} \pi_\alpha \cdot s_e(e'') \Bigr)
        \tag{\cref{proposition:least-solution-characterisation}} \\
    &\equivkao \sum_{e' \in \iota(e)} \Bigl( \epsilon(e') + \sum_{e'' \in \delta(e', a)} a \cdot e'' + \sum_{e'' \in \zeta(e', \alpha)} \pi_\alpha \cdot e'' \Bigr)
        \tag{\cref{proposition:least-solution-identity}} \\
    &\equivkao \sum_{e' \in \iota(e)} \epsilon(e') + \sum_{\substack{e' \in \iota(e) \\ e'' \in \delta(e', a)}} a \cdot e'' + \sum_{\substack{e' \in \iota(e) \\ e'' \in \zeta(e', \alpha)}} \pi_\alpha \cdot e''
        \tag{Distributivity} \\
    &\leqqkao \epsilon(e) + \sum_{e' \in \delta(e, a)} a \cdot e' + \sum_{e' \in \zeta(e, \alpha)} \pi_\alpha \cdot e'
        \tag{Def $\epsilon$, \cref{lemma:iota-derivatives-containment}}
     = \tilde{e}
       \ .
\end{align*}
Hence, $e \leqqkao \tilde{e}$, completing the proof.
\end{proof}

To prove \cref{nomorearrow}, the following intermediate lemma is useful.

\begin{restatable}{lem}{morethanonea}%
\label{morethanonea}
Let $\alpha \in \at$, $w\in\Gamma^\star$ and $e\in\termsgrat$.
Now, if $\alpha w \in \semkao{e}$, then there exists an $e' \in \zeta(e,\alpha)$ such that $w \in \semkao{e'}$.
\end{restatable}
\begin{proof}
By \cref{lemma:reconstruct-from-initial} and \cref{lemma:soundness}, we know that $\alpha{}w \in \semkao{\sum_{e' \in \iota(e)} e'}$.
Hence, by \cref{lemma:closure-almost-commutes}, there exists an $e' \in \iota(e)$ such that $\alpha{}w \in \semkao{e'}$.
By \cref{proposition:least-solution-identity} and \cref{lemma:soundness}, it then follows that $\alpha{}w \in \semkao{s_e(e')}$.
By \cref{proposition:component-closure}, we find that $\alpha{}w \in \semwkao{s_e(e')}$.
We furthermore know by \cref{proposition:least-solution-characterisation,lemma:wkao-soundness,lemma:atom-semantics}, as well as the definition of $\semwkao{-}$ that
\begin{align*}
\semwkao{s_e(e')}
    &= \bigsemwkao{\epsilon(e') + \sum\nolimits_{e'' \in \delta(e', a)} a \cdot s_e(e'') + \sum\nolimits_{e'' \in \zeta(e', \alpha)} \pi_\alpha \cdot e''} \\
    &= \semwkao{\epsilon(e')}
        \cup \bigcup_{e'' \in \delta(e', a)} \{ a \} \cdot \semwkao{s_e(e'')}
        \cup \bigcup_{e'' \in \zeta(e', \beta)} \{ \beta \} \cdot \semwkao{s_e(e'')}
        \ .
\end{align*}
Since $\alpha{}w$ starts with an atom, it appears in the third component above, i.e., there must exist $\beta \in \at$ and $e'' \in \zeta(e', \beta)$ with $\alpha{}w \in \{ \beta \} \cdot \semwkao{s_e(e'')}$.
Hence, $\alpha = \beta$ and $w \in \semwkao{s_e(e'')}$.
By \cref{proposition:component-closure,proposition:least-solution-identity,lemma:soundness}, we find that $\semwkao{s_e(e'')} = \semkao{e''}$.
Lastly, by \cref{lemma:iota-derivatives-containment}, we know that $e'' \in \zeta(e, \alpha)$.
In total, we conclude that $e'' \in \zeta(e, \alpha)$ such that $w \in \semkao{e''}$.
\end{proof}

\nomorearrow*
\begin{proof}
The inclusion from right to left follows by \cref{lemma:atom-semantics} and \cref{lemma:closure-almost-commutes}.

For the other direction, note that by \cref{lemma:atom-semantics} and \cref{lemma:wkao-soundness} we have
\[
    \bigsemwkao{\sum\nolimits_{e' \in \zeta(e, \alpha)} \pi_\alpha \cdot e'}
        = \bigcup_{e' \in \zeta(e, \alpha)} \{ \alpha \} \cdot \semwkao{e'}
        \subseteq \bigcup_{e' \in \zeta(e, \alpha)} \{ \alpha \} \cdot \semkao{e'}
        \ .
\]
It then suffices to prove that the right-hand side of the claimed equation is closed under $\contr$, since the left-hand side is the least set satisfying both of these conditions.
Thus, suppose $e' \in \zeta(e, \alpha)$, and let $w \in \{ \alpha \} \cdot \semkao{e'}$ with $x \contr w$.
We should find $e'' \in \zeta(e, \alpha)$ such that $x \in \{ \alpha \} \cdot \semkao{e''}$.
To this end, we note that $w = \alpha{}w'$ for some $w' \in \semkao{e'}$, and that moreover there exist $u, v \in \Gamma^\star$ and $\beta \in \at$ such that $w = u\beta\beta{}v$ and $x = u\beta{}v$.
This gives us two cases.
\begin{enumerate}[(i)]
    \item
    If $u = \epsilon$, then $\alpha = \beta$.
    Since $w' = \alpha{}v$, we find by \cref{morethanonea} an $e'' \in \zeta(e', \alpha)$ such that $v \in \semkao{e''}$.
    By \cref{lemma:double-derivatives}, we furthermore know that $e'' \in \zeta(e, \alpha)$.
    The claim is then satisfied, since $x = \alpha{}v \in \{ \alpha \} \cdot \semkao{e''}$.

    \item
    If $u \neq \epsilon$, then write $u = \alpha{}u'$.
    We then know that $w' = u'\beta\beta{}v$.
    By closure of $\semkao{e'}$ with respect to $\contr$, it follows that $u'\beta{}v \in \semkao{e'}$.
    Hence, $x = u\beta{}v = \alpha{}u'\beta{}v \in \{ \alpha \} \cdot \semkao{e'}$, satisfying the claim, since $e' \in \zeta(e, \alpha)$.
    \qedhere
\end{enumerate}
\end{proof}

\fundamentalsemantics*
\begin{proof}
We derive as follows.
\begin{align*}
\semkao{e}
    &= \bigsemkao{%
        \epsilon(e)
            + \sum\nolimits_{e' \in \delta(e, a)} a \cdot e'
            + \sum\nolimits_{e' \in \zeta(e, \alpha)} \pi_\alpha \cdot e'
        } \tag{\cref{lemma:soundness}, \cref{theorem:fundamental}} \\
    &= \semkao{\epsilon(e)}
        \cup \bigsemkao{\sum\nolimits_{e' \in \delta(e, a)} a \cdot e'}
        \cup \bigsemkao{\sum\nolimits_{e' \in \zeta(e, \alpha)} \pi_\alpha \cdot e'}
        \tag{\cref{lemma:closure-almost-commutes}} \\
    &= \{ \epsilon : \epsilon(e) = 1 \}
        \cup \bigsemkao{\sum\nolimits_{e' \in \delta(e, a)} a \cdot e'}
        \cup \bigsemkao{\sum\nolimits_{e' \in \zeta(e, \alpha)} \pi_\alpha \cdot e'}
        \tag{\cref{lemma:acc-vs-unit}} \\
    &= \{ \epsilon : \epsilon(e) = 1 \}
        \cup \bigcup_{e' \in \delta(e, a)} \{a\} \cdot \semkao{e'}
        \cup \bigcup_{e' \in \zeta(e, \alpha)} \{ \alpha \} \cdot \semkao{e'}
        \ .
        \tag{\cref{nomorearrow}}
\end{align*}
In the last step, we also use that closure distributes in the semantics of $a \cdot e'$.
As $a \in \Sigma$, we know that $\{a\} \cdot \semkao{e'} = \{a\} \cdot \down{\semwkao{e'}} = \down{(\{a\} \cdot \semwkao{e'})} = \down{\semwkao{a \cdot e'}} = \semkao{a \cdot e'}$.
\end{proof}

\semiota*
\begin{proof}
We derive as follows.
\begin{align*}
\semkao{e} = \semkao{f}
    &\Leftrightarrow \bigsemkao{\sum\nolimits_{e'\in\iota(e)}e'} = \bigsemkao{\sum\nolimits_{f'\in\iota(f)}f'}
        \tag{\cref{lemma:reconstruct-from-initial}} \\
    &\Leftrightarrow \bigcup_{e'\in\iota(e)}\semkao{e'} = \bigcup_{f'\in\iota(f)}\semkao{f'}
        \tag{\cref{lemma:closure-almost-commutes}} \\
    &\Leftrightarrow \bigcup_{e'\in\iota(e)}\ell(e') = \bigcup_{f'\in\iota(f)}\ell(f')
        \ .
        \tag*{(\cref{theorem:language-equiv}) \qedhere}
\end{align*}
\end{proof}

\decisionprocedure*
\begin{proof}
We derive as follows.
\begin{align*}
e\equivkao f
    & \Leftrightarrow \semkao{e}=\semkao{f} \tag{\cref{theorem:soundness-completeness}} \\
    & \Leftrightarrow  \bigcup_{e'\in\iota(e)}\ell(e') = \bigcup_{f'\in\iota(f)}\ell(f') \tag{\cref{cor:semiota}} \\
    & \Leftrightarrow \exists R\text{ such that } R \text{ is a bisimulation up to congruence} \\
    & \phantom{\Leftrightarrow } \text{ and } (\iota(e),\iota(f))\in R
        \ .
        \tag*{(\cref{theorem:bisim}) \qedhere}
\end{align*}
\end{proof}

\section{Tests versus observations}%
\label{appendix:tests-vs-obs}

We now investigate the relationships between models of \KAT and models of \KAO. 
To make this discussion more articulate, we note that both types of models are instances of a more general class of models: Kleene algebras $K$ containing a subset $B\subseteq K$ where $B$ is a Boolean algebra.
For the purpose of this discussion, we assume a priori two (and only these two) correspondences between the two signatures $\langle K, 0, 1, +, \cdot, \star \rangle$ and $\langle B, \bot, \top, \vee, \wedge, \overline{\cdot} \rangle$, namely:
\begin{enumerate}[(i)]
    \item
    that the natural order on $B$ is compatible with the natural order on $K$, that is, for all $a, b \in B$ such that $a \vee b = b$, it also holds that $a + b = b$.

    \item
    the contraction law, that is, for all $a, b \in B$ it holds that $a \wedge b \leq a \cdot b$.
\end{enumerate}

According to the laws of Boolean algebra, $\wedge$ is a greatest lower bound (GLB) for elements of $B$.
In most semantics, one would like to interpret conjunction as an intersection, hence as a GLB in $K$ (and not just in $B$).
Symbolically,
\begin{equation}
\forall x \in K,\ \forall a,b \in B,\ x \leq a \text{ and } x \leq b \implies x \leq a \wedge b \ . \tag{\textsc{Glb}}\label{eq:glb}
\end{equation}
Any \KAT satisfies~\eqref{eq:glb}, because for $x \in K$ and $a, b \in B$ with $x \leq a, b$, we have
\[
    x
        = x \cdot 1
        = x \cdot \left(b + \overline b\right)
        = x \cdot b + x \cdot \overline b \leq a \cdot b + b \cdot \overline b
        = a \cdot b + 0
        = a \wedge b
        \ .
\]

In most models of \KAT, including the language model and relational model, the Boolean sub-algebra is an ideal of the full algebra, i.e., it is downward-closed in the following sense:
\begin{equation}
 \forall x \in K,\ \forall a \in B,\ x \leq a \implies x \in B \ . \tag{\textsc{Ideal}}\label{eq:ideal}
\end{equation}
This is also the case in the language model of \KAO. 
For both algebras there are models that do not satisfy this property, but it is our impression that these are not the most useful.

\begin{restatable}{prop}{idealglb}%
\label{proposition:ideal-glb}
If~\eqref{eq:ideal} holds, then so does~\eqref{eq:glb}.
\end{restatable}
\begin{proof}
Let $x\in K$ and $a,b\in B$ such that $x\leq a,b$.
Since $a\in B$, we have $x \in B$ by~\eqref{eq:ideal}.
Therefore, we have $x = x\wedge x \leq a \wedge b$.
\end{proof}

\begin{restatable}{prop}{idealkao}%
\label{proposition:ideal-kao}
If~\eqref{eq:ideal} holds, then $K$ is a model of \KAO. 
\end{restatable}
\begin{proof}
First, we show that $a+b=a\vee b$.
Since $a+b\leqqba \top+\top\equivka\top$, we know that $a+b\in B$.
We also know that $a,b\leqqka a+b$ and that $a,b\leqqba a\vee b$, so we get:
$a\vee b \leqqka\left(a+b\right)\vee\left(a+b\right)\equivba a+b\leqqba\left(a\vee b\right)+\left(a\vee b\right)\equivka a\vee b$.

Now, since $0\leqqka \top$, $0$ is also an element of $B$, and therefore we can conclude:
\[
    0
        \equivba 0 \vee \bot
        = 0 + \bot\equivka \bot
        \ .
        \qedhere
\]
\end{proof}

\noindent
Consider now the following law, which holds trivially in any \KAT: 
\begin{equation}
1 = \top \ . \tag{\textsc{Units}}\label{eq:units}
\end{equation}
Since \KAO distinguishes $\cdot$ from $\wedge$, we distinguished their units as well.
This turns out to be important: identifying the units leads to identifying the products in a large class of models.

\begin{restatable}{prop}{glbunitprods}%
\label{proposition:glb-unit-prods}
  Suppose~\eqref{eq:glb} and~\eqref{eq:units} hold.
  Then for $a, b \in B$ we have $a \cdot b = a \wedge b$.
\end{restatable}
\begin{proof}
  Since we have assumed the contraction law, we need to check $a\cdot b\leq a \wedge b$ only.
  Thanks to~\eqref{eq:glb} this inequality reduces to checking that $a\cdot b\leq a$ and $a\cdot b\leq b$.
  Both of these hold, as for instance $a\cdot b\leq a \cdot \top= a \cdot 1 =a$.
\end{proof}

\begin{restatable}{prop}{kaovskat}
  Suppose~\eqref{eq:ideal} holds.
  Now $K$ is a \KAT iff~\eqref{eq:units} holds.
\end{restatable}
\begin{proof}
  The direct implication being straightforward, it remains to show the converse.
  By \cref{proposition:ideal-kao}, we know that $K$ is a model of \KAO, so it suffices to check that $\wedge$ and $\cdot$ coincide on $B$.
  Thanks to \cref{proposition:ideal-glb}, we know that~\eqref{eq:glb} holds, so we may conclude using \cref{proposition:glb-unit-prods}.
\end{proof}

This concludes our comparison between models of \KAT and models of \KAO. 

\fi%

\end{document}